\documentclass[11pt]{article} 
\usepackage[margin=1in]{geometry}

\usepackage[utf8]{inputenc}  
\usepackage[T1]{fontenc}    
\usepackage{hyperref}       
\usepackage{url}          
\usepackage{booktabs}        
\usepackage{amsfonts}        
\usepackage{nicefrac}        
\usepackage{microtype}       
\usepackage{authblk}
\usepackage{bm}
\usepackage{braket}
\usepackage{graphicx}
\usepackage{amsthm}
\usepackage{amssymb}
\usepackage{amsmath}
\usepackage[linesnumbered,ruled,vlined]{algorithm2e}

\usepackage{appendix}

\usepackage{etoc}

\usepackage{adjustbox}
\usepackage[normalem]{ulem}
\newtheorem{thm}{Theorem}
\newtheorem{lem}{Lemma}
\newtheorem{definition}{Definition}

\newtheorem{prop}{Proposition}

\newtheorem{fact}{Fact}
\newtheorem{remark}{Remark}
\usepackage{mathtools}
\usepackage{bbm}
\usepackage[explicit]{titlesec}
\usepackage{nccmath}
\usepackage{xcolor}
\usepackage{caption}
\usepackage{tkz-tab}
\usepackage{subcaption}
\usepackage{wrapfig}

\DeclareMathOperator{\Tr}{Tr}

\DeclareMathOperator{\RY}{R_Y}

\DeclareMathOperator{\Ber}{Ber}
\DeclareMathOperator{\Del}{Del}
\DeclareMathOperator{\Z}{Z}
\DeclareMathOperator{\X}{X}
\DeclareMathOperator{\PAC}{\mathsf{PAC}}
\DeclareMathOperator{\PPAC}{\mathsf{PPAC}}
\DeclareMathOperator{\SQ}{\mathsf{SQ}}
\DeclareMathOperator{\QSQ}{\mathsf{QSQ}}

\DeclareMathOperator{\Qstat}{\mathsf{Qstat}}

\title{On the learnability of quantum neural networks }

\author[1]{Yuxuan Du}
\author[2]{Min-Hsiu Hsieh}
\author[1]{Tongliang Liu}
\author[3]{Shan You}
\author[1]{Dacheng Tao}
\affil[1]{UBTECH Sydney AI Centre, School of Computer Science, Faculty of Engineering, The University of Sydney, Australia}
\affil[2]{Centre for Quantum Software and Information, Faculty of Engineering and Information Technology, University of Technology Sydney, Australia}
\affil[3]{SenseTime}

\date{}

\begin{document}

\maketitle

\begin{abstract}
We consider the learnability of the quantum neural network (QNN) built on the variational hybrid quantum-classical scheme, which remains largely unknown due to the non-convex optimization landscape, the measurement error, and the unavoidable gate errors introduced by noisy intermediate-scale quantum (NISQ) machines. Our contributions in this paper are multi-fold. First, we derive the utility bounds of QNN towards empirical risk minimization, and show that large gate noise, few quantum measurements, and deep circuit depth will lead to the poor utility bounds. This result also applies to the variational quantum circuits with gradient-based classical optimization, and can be of independent interest. We then prove that QNN can be treated as a differentially private (DP) model. Thirdly, we show that if a concept class can be efficiently learned by QNN, then it can also be effectively learned by QNN even with gate noise. This result implies the same learnability of QNN whether it is implemented on noiseless or noisy quantum machines. We last exhibit that the quantum statistical query (QSQ) model can be effectively simulated by noisy QNN. Since the QSQ model can tackle certain tasks with runtime speedup, our result suggests that the modified QNN implemented on NISQ devices will retain the quantum advantage. Numerical simulations support the theoretical results.   
\end{abstract}

\section{Introduction} 
Deep neural network (DNN) has substantially impacted the field of  machine learning in the past decade  \cite{goodfellow2016deep}. Most real-world applications, such as object detection \cite{he2017mask,lin2017feature}, question answering  \cite{devlin2018bert,yang2019xlnet}, social recommendation \cite{he2017neural}, among many others, could be accomplished by DNN-based learning algorithms with state-of-the-art performance because of  the powerful computational hardware and the flexible architecture of DNN. As shown in  Fig.~\ref{fig:QNN} (a), DNN adopts a multi-layer scheme. The inputs were processed through the feature embedding layers $\mathcal{F}_{\bm{x}}(\cdot)$, followed by the fully-connected layers $\prod_\ell W_\ell(\cdot)$, where the choice of each layer and the combination rule can be tailor made for various learning tasks.  Training DNN is a process to uncover the intrinsic relation between the input and the output of the given dataset. A huge amount of effort has been dedicated to understanding and explaining the  \textit{learnability} of DNN from the perspective of the convergence and the  generalization  \cite{allen2019convergence,backurs2017fine,golowich2017size,haussler1992decision,neyshabur2017exploring}; namely, the capabilities and limitations of DNN learning models.    

Quantum machine learning is a central application of quantum computing \cite{biamonte2017quantum}. With the aim of solving real-world problems beyond the reach of classical computers, firm and steady progress has been developed during the past decade \cite{ciliberto2018quantum,du2019efficient,dunjko2018machine,harrow2017quantum}. Among these breakthroughs, a quantum extension of DNN, i.e., the quantum neural network \textbf{(QNN)}, which is separately proposed in \cite{farhi2018classification,havlivcek2019supervised,mitarai2018quantum,schuld2019quantum},  received great attention due to the huge success of DNN and the superior computational power of quantum machines. As shown in Fig.~\ref{fig:QNN} (b), QNN also adopts the multi-layer architecture, where the inputs were converted into corresponding quantum states by the encoding quantum circuit $U_{\bm{x}}$, followed by trainable quantum circuits $U(\bm{\theta})=\prod_{l=1}^LU_l(\bm{\theta})$, where $\bm{\theta}$ is the adjustable parameter of quantum gates, and a classical optimizer. There is a close  correspondence between DNN and QNN: the feature embedding layers `$\mathcal{F}_{\bm{x}}$' of DNN coincide with the encoding quantum circuit $U_{\bm{x}}$ of QNN, while the fully-connected layer $W_l(\cdot)$ of DNN coincides with the trainable quantum circuit $U_l(\bm{\theta})$ of QNN. Celebrated by the stronger power of quantum circuits to prepare classical distributions \cite{aaronson2011computational,bremner2011classical}, QNN could  possess a stronger expressive power than its classical counterparts \cite{du2018expressive} and accelerates a wide range of  machine learning problems.

\begin{figure*}
\centering
\includegraphics[width=0.96\textwidth]{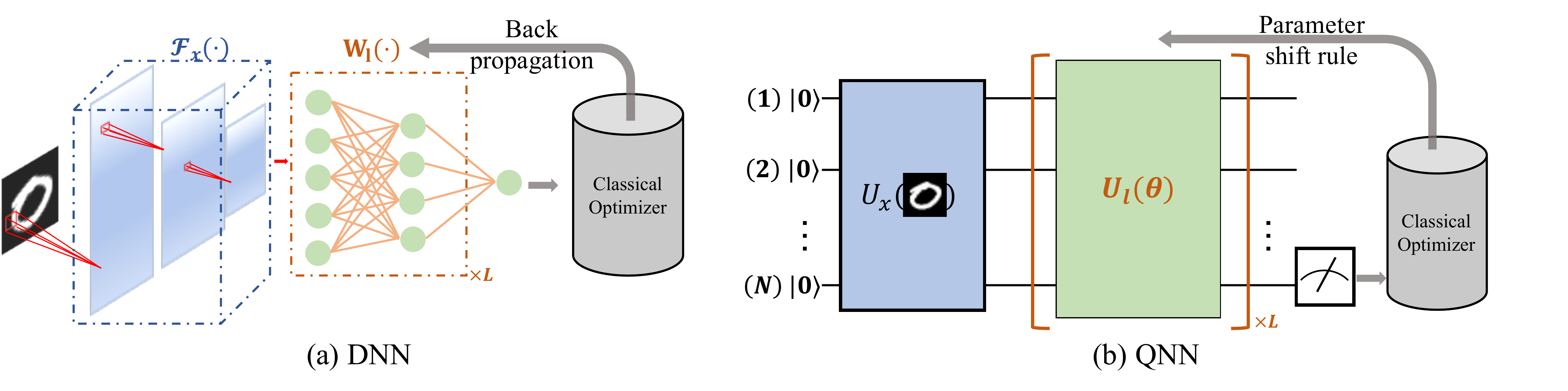}
\caption{\small{Illustration of DNN and QNN. The left and right panel shows DNN and QNN, respectively.  For DNN, the feature embedding layers $\mathcal{F}_{\bm{x}}(\cdot)$, which contains a sequence of operations with the arbitrary  combination such as convolution and attention, maps the input `0' to the feature space.   $W_l(\cdot)$ is the $l$-th fully-connected layer. For QNN, an encoding quantum circuit $U_{\bm{x}}$ maps the classical input  `0' to the quantum feature space. $U_l(\bm{\theta})$ is the $l$-th trainable quantum circuit. Classical information for optimization is extracted by quantum measurements.  }}
	\label{fig:QNN}
\end{figure*}  

Despite the promising prospects, theoretical results about QNN remain largely unknown. The \textbf{difficulties} mainly come from two sides. First,  the versatile structures of QNN and their non-convex optimization landscapes, similar to the DNN, heavily challenge the analysis. Second, due to the nature of quantum mechanics, the classical optimizer only receives estimated statistical information with a finite number of  measurements, and the error will pile up with the increased number of iterations. Although some studies have overcome partial difficulties from the aspect of vanishing gradients \cite{cerezo2020cost,mcclean2018barren},  robustness \cite{chakrabarti2019quantum,sharma2020noise}, information scrambling \cite{shen2020information}, memory capacity \cite{wright2019capacity}, and  no-free lunch theorem \cite{poland2020no}, the fundamental question, namely,  `\emph{What is the learnability of QNN}', is left open. 

The importance of exploring QNN's learnability is further increased in the noisy intermediate-scale quantum (NISQ) era \cite{arute2019quantum,preskill2018quantum}, since QNN can be easily built on NISQ machines and its performance is robust against gate noise. Empirical studies have shown that QNN can accomplish various supervised  learning tasks,  e.g., classification \cite{blank2020quantum,havlivcek2019supervised,schuld2019quantum}, regression \cite{mitarai2018quantum,killoran2018continuous}. However, no theoretical results can conclude any quantum advantage of these outcomes.  To theoretically explain the empirical observations, exploring the learnability of QNN under ERM framework \cite{vapnik2013nature} could be very fruitful, because ERM underpins many core results in statistical learning theory and offers learning guarantees for a wide range of supervised learning tasks.  Furthermore, it is unclear how gates noise affects the learnability of QNN. This answer substantially affects the feasibility of QNN on NISQ machines to purse quantum merits.

\textbf{Problem setup.}\label{sec:prob_setup} 
We follow the convention in statistical learning theory, and examine the learnability of QNN under the framework of empirical risk minimization (ERM) \cite{vapnik1992principles} as a first step.  In this way, analyzing the learnability of QNN amounts to checking the utility bounds generated by QNN. 	

Let $\bm{z}=\{\bm{z}_j\}_{j=1}^n \in \mathcal{Z}$ be the given dataset with $\mathcal{Z}$ being the sample domain, where the $i$-th sample  $\bm{z}_j=(\bm{x}_j,y_j)$ includes a feature vector $\bm{x}_j\in \mathbb{R}^D$ and a label $y_j\in \mathbb{R}$. ERM aims to find the optimal $\bm{\theta}^*\in\mathbb{R}^d$ by minimizing the objective function $\mathcal{L}$ within the constraint set $\mathcal{C}\subseteq \mathbb{R}^d$, i.e.,  
\begin{equation}\label{eqn:exp_erm}
	\bm{\theta}^*= \arg \min_{\bm{\theta}\in\mathcal{C}} \mathcal{L}(\bm{\theta},\bm{z}):=  \frac{1}{n}\sum_{j=1}^n \ell(y_i, \hat{y}_i) + r(\bm{\theta})~,
\end{equation}
where $\hat{y}_i$ is the predicted label that is determined by $\bm{\theta}$ and $\bm{x}_i$, $\ell$ is the loss function that measures the disparity between true labels $\{y_j\}_{j=1}^n$ and the predicted labels $\{\hat{y}_i\}_{i=1}^n$, and $r(\cdot)$ is a  regularizer. To ease the discussion, throughout the paper, we consider the square loss $\ell$, and use  $r(\bm{\theta})=\lambda\|\bm{\theta}\|_2^2/2$ with $\lambda \geq 0$. Note that our analysis can be {easily generalized} to other loss functions $\ell$ that satisfy $S$-smooth and $G$-Lipschitz properties  as discussed in Sec. \ref{sec:learn_QNN}.  

The common optimization rule to tackle ERM is the batch gradient descent method \cite{goodfellow2016deep}. Depending on the available resources,  the sample indices are  divided into $B$ disjoint batches $\{\mathcal{B}_i\}_{i=1}^B$ with equal size $B_s$, namely, $\bm{z}=\cup_{j\in \{\mathcal{B}_i\}_{i=1}^B} \bm{z}_j$. The optimization rule at the $t$-th iteration is $\bm{\theta}^{(t+1)}=\bm{\theta}^{(t)}-\frac{\eta}{B}\sum_{i=1}^B \nabla \mathcal{L}(\bm{\theta}^{(t)}, \mathcal{B}_i)$, where $\eta$ is the learning rate, the gradient $\nabla \mathcal{L}(\cdot)$ is
\begin{equation}\label{eqn:grad_QNN}
\nabla \mathcal{L}(\bm{\theta}^{(t)},\mathcal{B}_i)= \left(\hat{Y}_i^{(t)}-Y_i \right)\frac{\partial \hat{Y}_i^{(t)}}{\partial \bm{\theta}^{(t)}}  + \lambda \bm{\theta}^{(t)}~,	
\end{equation}
$Y_i=\frac{1}{B_s} \sum_{j\in\mathcal{B}_i}y_j$ and $\hat{Y}_i^{(t)}=\frac{1}{B_s}\sum_{j\in\mathcal{B}_i}\hat{y}_j^{(t)}$ are the sum average of the true labels and the predicted labels for the $i$-th batch $\mathcal{B}_i$, respectively. 
When no confusion will occur, we use $ \mathcal{L}(\bm{\theta}^{(t)})$ and $  \mathcal{L}_i(\bm{\theta}^{(t)})$ instead of $ \mathcal{L}(\bm{\theta}^{(t)},\bm{z})$ and $\mathcal{L}(\bm{\theta}^{(t)},\mathcal{B}_i)$ in the rest of study.

The training of QNN is similar to those of DNN. In particular, QNN also generated a sum average of the predicted labels, based on $\bm{\theta}$ and $\mathcal{B}_i$, after the measurement component in Fig.~\ref{fig:QNN} (b). However, the major difference between the gradient-based optimization of QNN and DNN is as follows. In DNN, the gradient in Eqn.~(\ref{eqn:grad_QNN}) can be easily obtained via backpropagation  \cite{goodfellow2016deep}. However, due to the nature of quantum mechanics, the gradient of a quantum unitary operator (e.g., trainable quantum circuit layer $U_l(\bm{\theta})$) is, in general, not a legitimate quantum operator anymore \cite{schuld2019evaluating}. To overcome this shortcoming, the \textit{parameter shift rule} \cite{mitarai2018quantum,schuld2019evaluating} is proposed to estimate the gradients of a quantum unitary operator using $K$ measurements. However, difficulties arise since only approximated $\hat{Y}_i^{(t)}$ and $\partial \hat{Y}_i^{(t)}/\partial \bm{\theta}^{(t)}$ are available due to a finite number of measurements, and the precision deteriorates when more iterations occur. The detailed steps will be discussed in Sec.~\ref{sec:learn_QNN}.

Furthermore, we would like to incorporate the unavoidable gate noise of the trainable quantum circuit $U(\bm{\theta})$ in our studies.  This can be done by considering the worst-case scenario, i.e., modeling the gate noise at each circuit depth to be  quantum {depolarization noise} $\mathcal{N}_p$ \cite{nielsen2010quantum}.  Intuitively, if a quantum state passes through $\mathcal{N}_p$, with probability $1-p$, the output remains unchanged; otherwise, all information of the input is lost and the output is the maximally mixed state. Note that the achieved results can be {easily extended} to a more general noisy channel (See Appendix K for details).

We adopt two standard utility metrics to quantify the performance (learnability) of QNN: 
\begin{equation}
R_1(\bm{\theta}^{(T)},\bm{z}) :=  \mathbb{E}\left[\|\nabla \mathcal{L}(\bm{\theta}^{(T)},\bm{z})\|^2\right],~R_2(\bm{\theta}^{(T)},\bm{z}) := \mathbb{E}[\mathcal{L}(\bm{\theta}^{(T)},\bm{z})] - \mathcal{L}(\bm{\theta}^*,\bm{z})~,
\end{equation}   
where $\bm{\theta}^{(T)}$ is   the output of QNN after $T$ iterations and $\nabla \mathcal{L}(\cdot)$ denotes the gradient of the function $\mathcal{L}(\cdot)$. The metric $R_1$ evaluates how far QNN is away from the stationary point, $\|\nabla \mathcal{L}(\bm{\theta}^{(T)},\bm{z})\|^2=0$, in expectation \cite{koltchinskii2011oracle,zhang2017efficient}. The utility metric $R_2$ evaluates the expected excess empirical risk \cite{bartlett2006convexity,bartlett2006empirical}. Due to the hardness to find the global optima in the non-convex landscape, $R_2$ can only be applied to some special non-convex objective functions, i.e., the objective functions satisfy the Polyak-Lojasiewicz (PL) condition \cite{nesterov2006cubic,wang2019differentially}. We will show that, under a mild assumption, the objective function of QNN also meets the PL condition,  and $R_2$ can be employed to analyze its performance.

\textbf{Contributions.} The main contributions of this study are as follows. Our first contribution is deriving  QNN's utility bounds for ERM. As aforementioned, the non-convex optimization landscape, the piled up estimation error due to quantum measurements and the inevitable  gate noise, heavily challenge the analysis of QNN's utility bounds. To the best of our knowledge, this is the first study towards understanding the learnability of QNN with the provable guarantee.    
\begin{thm} \label{thm:informal_utl_QNNQAE_DP}
QNN outputs $\bm{\theta}^{(T)}\in\mathbb{R}^d$ after $T$ iterations with utility bounds  \[R_1 \leq  \tilde{O}\left(d, \frac{1}{\frac{1}{BK}K}, \frac{1}{(1-p)^{L_Q}} \right)~\text{and,}~R_2\leq \tilde{O}\left(\frac{1}{K^2B}, d, \frac{1}{(1-p)^{L_Q}} \right)~,\] where $K$ is the number of quantum measurements, $L_Q$ is quantum circuit depth, $p$ is the gate noise, and $B$ is the batch size.
\end{thm} 
Theorem \ref{thm:informal_utl_QNNQAE_DP} indicates that a larger number of measurements $K$, a smaller gate noise rate $p$, a shallower circuit depth $L_Q$, and a smaller number of trainable parameters $d$ can yield a better utility bounds for both $R_1$ and $R_2$. We remark that the achieved utility bounds $R_1$ and $R_2$ are very general, and cover various types of encoding quantum circuits $U_{\bm{x}}$ and trainable quantum circuits $U(\bm{\theta})$. In particular,  our  results cover all typical encoding circuits, e.g., amplitude encoding \cite{plesch2011quantum,schuld2017implementing,schuld2020circuit}, kernel mapping \cite{havlivcek2019supervised,mitarai2018quantum,schuld2019quantum}, dimension reduction method \cite{wilson2018quantum}, and basis encoding methods \cite{kapoor2016quantum,farhi2018classification}, and a diverse architectures of the trainable quantum circuit, as long as it is composed of the parameterized single qubit gates and two qubits gates \cite{benedetti2019parameterized}. 
 
Note that the variational hybrid quantum-classical learning models have also been empirically applied to explore fundamental properties of physical systems, e.g., ground energies approximation  and thermal averages computation \cite{motta2020determining,peruzzo2014variational}. These problems are generally more sensitive to the global minimum than that of machine  learning problems. Therefore the utility bounds in Theorem \ref{thm:informal_utl_QNNQAE_DP} can serve as a powerful tool to support those results.

A central topic in classical machine learning is exploring whether the noise affects the learnability of a given learning task. A notable example is that the class of parity functions is  probably approximately correctly ($\PAC$)  learnable; however, learning parity with noise is thought to be computationally hard \cite{blum2003noise}.  Here we lift this essential question from the classical scenario to the quantum scenario: whether there exists any concept class that separates the learnability of the noiseless QNN with noisy QNN, i.e., noiseless QNN uses polynomial samples to learn this concept class, while exponential samples are needed for the NISQ case.

Our second contribution is providing a negative answer towards the above question.
   \begingroup
   \allowdisplaybreaks
   \begin{thm}\label{thm:DP_QNNQAE_inform} 
If QNN with noiseless gates $\PAC$ learns a concept, then there exists a modified  QNN with certain types of noisy gates that can also learn this concept using polynomial samples.     
\end{thm}
\endgroup
The result of Theorem \ref{thm:DP_QNNQAE_inform} indicates the same sample complexity between noiseless QNN  and noisy QNN to learn a specific concept class. This implies that if QNN achieves certain learning tasks with quantum advantages, then we can implement QNN on NISQ machines with a simple  modification to preserve advantages as well.   

The key technique used to achieve Theorems \ref{thm:DP_QNNQAE_inform} is differentially private (DP) learning    \cite{chaudhuri2011differentially,dwork2014algorithmic,ying2017quantum,zhou2017differential}. The intuition to employ DP is as follows. The behavior of QNN with gate noise resembles  learning with noise and DP learning, where a certain type of  noise is injected into the learning model. However, DP learning dispels learning with noise \cite{kasiviswanathan2011can}, where the former can effectively tackle some tasks that are computationally hard for the latter.   Hence, it is beneficial to explore whether QNN with gate noise belongs to a DP learning model instead of learning with noise.  The  exploration about the DP property of QNN  leads to our third contribution. 
\begin{lem} \label{lem:DP_QNN01}
The QNN with gate noise can be treated as a $(\epsilon, \delta)$-DP model with $\delta\geq 0$ and  \[\epsilon=\tilde{O}\left(\sqrt{Td}+Td \left(\frac{(1-\tilde{p})+\tilde{p}\frac{\Tr(\Pi)}{D}}{\left(\tilde{p}(1-\tilde{p})(1-\frac{\Tr(\Pi)}{D})\right)^K}\right)^d - Td\right)~.\]
\end{lem}
Together with the fact that non-private and DP algorithms share the same learnability in terms of sample complexity \cite{kasiviswanathan2011can}, we complete Theorem \ref{thm:DP_QNNQAE_inform}.

Last, we explore the learnability of QNN implemented on NISQ machines. In particular, we aim to find certain tasks that can be achieved by these two learning models with quantum advantages. To reach this goal, we explore whether quantum statistical query learning ($\QSQ$) model can be efficiently simulated by these two learning models, since  $\QSQ$ can efficiently tackle parity learning, juntas learning, and  DNF (disjunctive normal form)  learning problems with quantum advantages, whereas these problems are computationally hard for classical $\SQ$ models \cite{arunachalam2020quantum}.  Our third contribution is exhibiting that $\QSQ$ can only be efficiently simulated by QNN with gate noise, and establish a computational separation between the original QNN and modified QNN.  
 \begin{thm}\label{thm:QNN_QAE_QSQ_info}
A $\QSQ$ learning model can be efficiently simulated by  QNN with gates noise using polynomial samples. 
\end{thm}

 \subsection{Related work}
 Previous quantum machine learning literatures that are related to our work can be divided into two groups: quantum learning theory and quantum neural networks.  We address that, none of the  studies listing below have concerned the learnability of QNN.

 For the first group, the studies \cite{arunachalam2017guest,atici2005improved,bernstein1997quantum,servedio2004equivalences} exhibited that the sample complexity of quantum and classical probably approximately correct learning ($\PAC$) (or agnostic) learning models is equal up to a constant factor under the distribution-independent setting. A recent  study \cite{arunachalam2020quantum} generalized the classical statistical query model ($\SQ$) to the quantum statistical query ($\mathsf{QSQ}$)  model and  compare the learnability among $\SQ$ learner, (noisy) $\text{quantum} \PAC$ learner, $\QSQ$ learner, and $\PPAC$ learner. However, how to use these results to analyze the  learnability of QNN is inexplicit. 
 
 For the second group, beyond the hybrid scheme as discussed in this study, there are different schemes and platforms to implement QNN.  Specifically, several studies have investigated how to implemented QNN on noiseless quantum machines \cite{beer2020training,Kerenidis2020Quantum,wiebe2014quantum}, quantum reservoir \cite{ghosh2019quantum}, and quantum  annealers \cite{Bartlomiej2018}.  Since these proposals adopt distinct frameworks, they are incomparable with our results.

\section{Preliminary}
We unify the notations throughout the whole paper. We denote $D$ as the feature dimension ($\bm{x}\in\mathbb{R}^D$), $d$ as the number of training parameters ($\bm{\theta}\in\mathbb{R}^d$). Define $N$ as the number of qubits  and $n$ as the number of training examples. Denote the set $\{1,2,...,m\}$ as $[m]$. A random variable $X$ that follows Bernoulli distribution is denoted as $X\sim \Ber(p)$, i.e., $\Pr(X=1)=p$ and $\Pr(X=0)=1-p$.   With a slight abuse of notations, we denote $\ell_b$ as the $b$-norm, while $\ell$ (without subscript) is the loss function. We use $O(\cdot )$ (or $\tilde{O}(\cdot )$) to denote the complexity bound  (hide poly-logarithmic factors). See Appendix \ref{appd:Qc_pre} for details.

\textbf{Quantum computing.} 
	 We show basic insights of quantum computing.  	 
	
Quantum state works in the Hilbert space $\mathcal{H}$ with $\mathcal{H}\approxeq \mathbb{C}$. Let  $\ket{0} = \big(\begin{smallmatrix}
  1\\
 0
\end{smallmatrix}\big)$ and $\ket{1} = \big(\begin{smallmatrix}
  0\\
  1
\end{smallmatrix}\big)$ be the
standard \textit{basis states} for $\mathbb{C}^2$. A quantum bit (\textbf{qubit}) lives in a two-dimensional Hilbert space formed by $\ket{0}$ and $\ket{1}$. Multiple qubit basis states follow the tensor products rule, e.g.,  $\ket{0}\otimes \ket{1}\equiv\ket{1}\ket{0} \in \mathbb{C}^4$ describes a basis state of a 2-qubit system.  A \textit{pure state} $\ket{\bm{a}}$ with $N$-qubits follows  $\ket{\bm{a}} =  \sum_{i=1}^{d} \bm{a}_i \ket{i}$ with $d=2^N$ and $\|\bm{a}\|^2=1$, where the basis state $\ket{i}\in \{\ket{0},\ket{1}\}^{\otimes N}$ is also called \textit{computation basis}. $\ket{\bm{a}}$ is in \textit{superposition} if  $\|\bm{a}\|_{0}>1$.   The conjugate transpose of $\ket{\bm{a}}$ is denoted as $\bra{\bm{a}}$.    We use  \textit{density matrix} to describe more general  quantum states.  Given a mixture of $m$  pure states $\{p_i, \ket{\psi_i}\}_{i=1}^m$ with $p_i\geq 0$ and $\sum_{i=1}^m p_i =1$,  the  density matrix $\rho$ is $\rho = \sum_{i=1}^m p_i\rho_i$ with $\rho_i =\ket{\psi_i}\bra{\psi_i}$ and $\Tr(\rho)=1$.   There are two main types of quantum operations in quantum computation. The first one is \textit{quantum channel}, which is a completely positive trace-preserving map, e.g., applying a channel $\mathcal{N}$ to a density matrix $\rho\in\mathbb{C}^{d\times d}$ generates the state $\mathcal{N}(\rho) =\sum_a \mathbf{M}_a\rho \mathbf{M}_a^{\dagger}$ with $\sum_a\mathbf{M}_a\mathbf{M}_a^{\dagger}=\mathbb{I}_d$. Note that, a quantum gate, which is a unitary  matrix, is a special quantum channel. The second one is \textit{quantum measurement}, which extracts classical information from quantum state. An $m$-outcome measurement, a.k.a. positive-operator-valued measure (\textbf{POVM}), is modeled by $m$ positive semidefinite matrices $\{\Pi_b\}_{b=1}^m$ with $\sum_b \Pi_b=\mathbb{I}_d$. Given $\rho$, the probability to get outcome $b$ is $p_b = \Tr(\Pi_b\rho)$.  

\begin{definition}[Depolarization channel]\label{def:dp_channel}
Given a quantum state $\rho$, the depolarization channel $\mathcal{N}_p$ acts on $D$-dimensional Hilbert space is defined as $	\mathcal{N}_p(\rho) = (1-p)\rho + p{\mathbb{I}}/{D}$.
\end{definition}

\textbf{Differential privacy.} Differential privacy (DP) is a rigorous and standard notion for data privacy, which aims to train  an accurate learning model  without  exposing the precise information in individual training example, e.g., genomic data and  medical records for patients 
\cite{dwork2014algorithmic}.  
\begin{definition}[$(\epsilon,\delta)$-DP] \label{def:CDP}
Let $h(\bm{z},\bm{z}')$ be the hamming distance.  An algorithm $\mathcal{M}$ is $(\epsilon, \delta)$-differential private if for any two datasets $\bm{z},\bm{z}'\in \mathcal{Z}$ and a distance measure  $h(X,X')\leq 1$, and for all measurable sets $\mathcal{O}\subseteq \text{Range}(\mathcal{M})$, the following holds.  $\Pr(\mathcal{A}(\bm{z}) \in \mathcal{O}) \leq e^{\epsilon} \Pr(\mathcal{M}(\bm{z}')\in \mathcal{O})+\delta.$
\end{definition}

\section{Utility bounds of quantum neural network  towards ERM}\label{sec:learn_QNN}
A well-known consequence in ERM study is that the utility bounds of a given learning model massively depend on what kind of  and how much error contained in its gradient  \cite{bassily2014private,chaudhuri2011differentially,mei2018}. Specifically, when the gradient is perturbed by a sufficiently large amount  of noise, the optimization may not  converge and the utility bound is poor \cite{boyd2004convex}. Meanwhile,  empirical and theoretical evidence has corroborated that, injecting certain types of noise into the gradient does not affect or can even accelerate the   convergence \cite{chaudhuri2011differentially,li2017convergence,zhou2019towards}. A similar issue also happens to optimize QNN. In particular, when QNN is realized on quantum chips, the presence of sampling error and gate error enables that the classical optimizer only has access to an estimated gradient instead of the analytic gradient. However, theoretical results about how the involved estimation error of gradient  affects the optimization remain largely unknown. Moreover, the heuristic study  \cite{sung2020exploration} showed that  the conclusions based on certain quantum learning models, which are built under the ideal setting that omits the gate error or sample error, may not be applicable to experiments.    Therefore, it is crucial to establish the analytical relation between the estimated and analytic gradients, since this  relation is not only the precondition to analyze the utility bounds of QNN towards ERM, but can also   be used to quantify how the hybrid classical-quantum learning schemes perform on real quantum devices as an independent interest.  

We first elaborate the workflow of QNN. As shown in Figure  \ref{fig:QNN} (b), QNN first employs a state preparation unitary $U_{\bm{x}}$ to encode classical inputs $\{\bm{x}_j|j\in \mathcal{B}_i\}$ into quantum states, followed by the quantum circuit $U(\bm{\theta})$ with tunable parameter $\bm{\theta}$ to produce the state $\gamma_{\mathcal{B}_i} $. We refer the interested reader to Appendix \ref{appen:QNN} for implementation details of $U_{\bm{x}}$ and $U(\bm{\theta})$. Finally, a two-outcome measurement POVM $\Pi$ is applied to the state $\gamma_{\mathcal{B}_i} $  and produces the outcome $V_i $ that can be viewed as a binary random variable with the Bernoulli distribution $\Ber(\hat{Y}_i )$, where $\hat{Y}_i  :=\Tr(\Pi \gamma_{\mathcal{B}_i} )$. Denote the obtained statistics, i.e., the sample mean, by $\bar{Y}_i =\frac{1}{K}\sum_{k=1}^KV_k $ after repeating the above procedure $K$ times. The law of quantum mechanics ensures $\bar{Y}_i  \to \hat{Y}_i $ when $K\to \infty$. However, in reality, only a finite number of measurements is allowed, and this results in the sample error (measurement error).

In addition, the quantum gates in NISQ machines, which are used to implement $U_{\bm{x}}$ and $U(\bm{\theta})$, are prone to having errors \cite{preskill2018quantum}. The gate noise can be modelled by applying certain quantum channels to each quantum circuit layer, and we use the depolarization channel $\mathcal{N}_p$ in Definition \ref{def:dp_channel} in the following analysis. Note that our analysis works for more general channels, as discussed in Remark \ref{remark:other_channel_qnn}.  Specifically, with applying $\mathcal{N}_p$ to each layer of quantum circuit, the quantum state before measurement is $\tilde{\gamma}_{\mathcal{B}_i}=\mathcal{N}_p ( \gamma_{\mathcal{B}_i}) $ instead of  $\gamma_{\mathcal{B}_i} $. When the POVM $\Pi$ is applied to the state $\tilde{\gamma}_{\mathcal{B}_i} $, the obtained outcome $V_i $ follows the distribution $\Ber({\tilde{Y}_i} )$ with $\tilde{Y}_i  :=\Tr (\Pi \tilde{\gamma}_{\mathcal{B}_i} )$ instead of $\Ber({\hat{Y}_i} )$.

Recall that the updating rule of QNN  at the $t$-th iteration: $\bm{\theta}^{(t+1)}=\bm{\theta}^{(t)}-\frac{\eta}{B}\sum_{i=1}^B \nabla \mathcal{L}_i(\bm{\theta}^{(t)})$, requires the computation of the gradient 
$\nabla_j \mathcal{L}_i(\bm{\theta}^{(t)})=(\hat{Y}_i^{(t)}-Y_i){\partial \hat{Y}_i^{(t)}}/{\partial \bm{\theta}_j^{(t)}}  + \lambda \bm{\theta}_j^{(t)}$ with $j\in[d]$. 
In order to obtain the gradient $\nabla_j \mathcal{L}(\bm{\theta}^{(t)})$, the parameter shift rule is developed \cite{mitarai2018quantum,schuld2019evaluating}, since the gradient of a quantum unitary operator may not be  a legitimate quantum operation and cannot be realized on quantum circuits.  Specifically, the parameter shift rule proceeds by  separately feeding tunable parameters $\bm{\theta}^{(t)}$ and  $\bm{\theta}^{(t,\pm_j)}:=\bm{\theta}^{(t)}\pm \frac{\pi}{2}\bm{e}_j$ to the trainable circuit  $U(\bm{\theta})$, where $\bm{e}_j$ is the basis vector with  the $j$-th entry being $1$ and zero otherwise. Following the notations used above, we denote $\hat{Y}_i^{(t)}$ and $\hat{Y}_i^{(t,\pm_j)}$ as the expectation values of quantum measurements  when feeding parameters $\bm{\theta}^{(t)}$ and  $\bm{\theta}^{(t,\pm_j)}$ into trainable quantum circuit $U(\bm{\theta})$ in the noiseless scenario. The corresponding analytic  gradient is \[\nabla_j \mathcal{L}_i(\bm{\theta}^{(t)})=(\hat{Y}_i^{(t)}-Y_i)\frac{\hat{Y}_i^{(t,+_j)}-\hat{Y}_i^{(t,-_j)}}{2} + \lambda \bm{\theta}_j^{(t)}~.\]  However, in practice, QNN could only generate statistics $\bar{Y}_i^{(t)}= \frac{1}{K}\sum_{k=1}^K V_k^{(t)}$ and $\bar{Y}_i^{(t,\pm_j)}= \frac{1}{K}\sum_{k=1}^K V_k^{(t,\pm_j)}$, where $V_k^{(t)}\sim \Ber(\tilde{Y}_i^{(t)})$ and $V_k^{(t,\pm_j)}\sim \Ber(\tilde{Y}_i^{(t,\pm_j)})$, and $\tilde{Y}_i^{(t)}$ and $\tilde{Y}_i^{(t,\pm j)}$ refer to the expectation values of quantum measurements when feeding parameters $\bm{\theta}^{(t)}$ and  $\bm{\theta}^{(t,\pm_j)}$ into the noisy trainable quantum circuit $U(\bm{\theta})$. This leads to the estimated gradient as  \[\nabla_j \bar{\mathcal{L}}_i(\bm{\theta}^{(t)})=(\bar{Y}_i^{(t)}-Y_i)\frac{\bar{Y}_i^{(t,+_j)}-\bar{Y}_i^{(t,-_j)}}{2} + \lambda \bm{\theta}_j^{(t)}~.\]  

Our main technical contribution here is showing that the estimated gradient, which is caused by the gates noise and the sampling error, can be related to its optimal gradient, and can be explicitly formulated. An informal  result is summarized below (See Theorem D.1 in Appendix \ref{appendix:subsec:QNN_grad_dist} for details).
\begin{thm}\label{thm:noise_QNN_gaussian}
It follows that 
$$\nabla_j\bar{\mathcal{L}}_i(\bm{\theta}^{(t)}) = (1-\tilde{p})^2\nabla_j{\mathcal{L}}_i(\bm{\theta}^{(t)}) + C_{j,1}^{(i,t)} +  \bm{\varsigma}_i^{(t,j)},$$ where $\tilde{p}=1-(1-p)^{L_Q}$, $L_Q$ is the circuit depth, the constant $C_{j,1}^{(i,t)}$ only depends on ${Y}_i$, $\bm{\theta}^{(t)}$, and $\tilde{p}$, and  $\bm{\varsigma}_i^{(t,j)}$ follows the distribution $\mathcal{P}_{Q}$ that is formed by ${Y}_i$, $\bm{\theta}^{(t)}$, the number of measurements $K$, and $\tilde{p}$ with zero mean.
\end{thm}
  
The achieved result in Theorem \ref{thm:noise_QNN_gaussian} indicates that the estimated  gradient $\nabla_j\bar{\mathcal{L}}_i(\bm{\theta}^{(t)})$ is centralized around the $(1-\tilde{p})^2\nabla_j{\mathcal{L}}_i(\bm{\theta}^{(t)}) + C_{j,1}^{(i,t)}$ and perturbed by a random variable $\bm{\varsigma}_i^{(t,j)}$. This enables us to quantitively measure how far the estimated gradient is away from the analytic gradient, which is the precondition to leverage the optimization theory to analyze the performance of QNN. Moreover, the result of Theorem \ref{thm:noise_QNN_gaussian} implies that, compared with the finite measurements, the gate error is more harmful for the QNN's optimization, which may lead to diverging. In particular,    the term $C_{j,1}^{(i,t)}$, which is independent with $K$, will always exist and induce a biased optimization direction when $\tilde{p}\neq 0$. For the worst case, with $\tilde{p}=1$, the analytic gradient information is exactly lost. In contrast, $K$ only determines the variance of the distribution $\mathcal{P}_{Q}$ with zero mean, where  classical and quantum literatures \cite{sweke2019stochastic,zhou2019towards} have provided the convergence guarantee even if $K=1$.

Beside the effects of gradient error, the utility bounds also heavily depend on the properties of objective functions. In the following, we show that $\mathcal{L}$ used in QNN satisfies $S$-smooth, $G$-Lipschitz, and PL condition. The formal definitions of these concepts and the achieved results are given below. 
\begin{definition}\label{def:S-smoo-G_lip}
A function f is $S$-smooth over a set $\mathcal{C}$ if $\nabla^2 f(\bm{u})\preceq S\mathbb{I}$ with $S>0$ and $\forall \bm{u}\in \mathcal{C}$. A function f is  $G$-Lipschitz over a set $\mathcal{C}$ if for all $\bm{u},\bm{w}\in \mathcal{C}$, we have $|f(\bm{u}) - f(\bm{w})|\leq G\|\bm{u}-\bm{w}\|_2$. A function $f$ satisfies PL condition if there exists $\mu>0$ and for every possible $\bm{\theta}\in\mathcal{C}$,   $\|\nabla f(\bm{\theta}) \|^2 \geq 2\mu(f(\bm{\theta})-f^*)$, where $f^*=\min _{\bm{\theta}\in\mathcal{C} } f(\bm{\theta})$.
\end{definition}

\begin{lem}\label{lem:Lsmmoth}
Following the notations in  Eqn.~(\ref{eqn:exp_erm}), $\mathcal{L}(\bm{\theta})$ is $S$-smooth with $S = (\frac{3}{2}+ \lambda)$ and $G$-Lipschitz with $G=d(1+3\pi\lambda)$.  Assuming  $\lambda > \frac{1}{\pi}$,  $\mathcal{L}$  satisfies  PL condition with $\mu = \frac{(-1+ \lambda \pi)^2}{1+{\lambda d}(3\pi)^2}$.
 \end{lem}  
 The proof of this lemma is provided in Appendix \ref{Append:obj_func_prop}. 
 
The analysis of the utility bounds $R_1$ and $R_2$ of QNN towards ERM, which are  summarized in Theorem \ref{thm:informal_utl_QNNQAE_DP}, can be effectively conducted by leveraging Theorem \ref{thm:noise_QNN_gaussian} and Lemma  \ref{lem:Lsmmoth}.  Theorem \ref{thm:informal_utl_QNNQAE_DP}   provides the following theoretical guidances to design QNN-based learning algorithms, i.e., a larger amount of measurements $K$ and lager batch size $B$, smaller deporlarizing error $p$, smaller parameter space $d$, and shallower quantum circuit $L_Q$, can yield a better utility bounds $R_1$ and $R_2$.   

The full proof of  Theorem \ref{thm:informal_utl_QNNQAE_DP} is provided in Appendix \ref{Appendix:Thm_utl_QNN}.  The proof strategy of Theorem \ref{thm:informal_utl_QNNQAE_DP} is as follows. Recall that the utility bound $R_1$ measures how far the trainable parameter of QNN is away from the stationary point.  A well-known result in optimization theory  \cite{jin2017escape} is that the  stationary point of  a function can be efficiently located by a simple analytic gradient-based algorithm, once the function satisfies the smooth property. Hence, to achieve $R_1$,  we can utilize the smooth property of $\mathcal{L}$ and the result of Theorem \ref{thm:noise_QNN_gaussian}, which    reformulates the estimated gradient by the analytic gradient, to analyze the stationary convergence of QNN. The key component to achieve $R_2$ is the PL condition. Recall that the utility bound $R_2$ evaluates the disparity between the expected and the optimal empirical risk. The study \cite{nesterov2006cubic} indicates that, if a non-convex function satisfies PL condition, then every stationary point is the global minimum. Alternatively, PL relates the stationary point with the optimal empirical risk, which is determined by the global minimum. Hence, by leveraging the PL condition and the result of $R_1$, we can obtain the utility bounds of $R_2$.  
 
\begin{remark}\label{remark:other_channel_qnn}
Theorem \ref{thm:informal_utl_QNNQAE_DP}   can be easily extended to a more general noisy channel $\mathcal{E}_{p_1}$, i.e., $\mathcal{E}_{p_1}(\rho)=(1-p_1)\rho + p_2 \kappa + p_3\mathbb{I}_{D}/D$, where $\rho,\kappa \in \mathbb{C}^{D\times D}$ and $p_2+p_3=p_1$. See Appendix \ref{Appendix:sec_more_general_channel} for details. 
\end{remark}
  
\section{The learnability of quantum neural networks with noisy gates}\label{sec:learnability_QNN_QAE}
In the analysis of ERM, we exhibit that the gate noise of QNN massively affects its utility bounds. In this section, we aim to understand how noise in QNN affects its learning capabilities in terms of the sample complexity; namely, whether any concept class that can be probably approximately correctly ($\PAC$) learned by QNN with noiseless gates can also be $\PAC$ learned by QNN with noisy gates. If the answer is negative, this concept class  is unlikely to be efficiently learned on the NISQ quantum devices. Moreover, it will demonstrate the inequivalent learnability of  noiseless QNN and noisy QNN.  In the classical literature, the class of parity functions serves as an excellent example to separate the learnability of the $\PAC$ learning model with the statistical query ($\SQ$) learning model  \cite{blum2003noise}. Furthermore, there is an even more pressing need to understand what kinds of concept classes  can be efficiently learned by  QNN with quantum advantages. Towards  this question, we explore whether any concept class that is learnable in the quantum statistical query ($\QSQ$) model \cite{arunachalam2020quantum} is also learnable by noiseless and noisy QNN, enlighten by the fact that $\QSQ$ model can tackle certain learning tasks that outperform its classical counterpart.

In order to answer the above questions, we attempt to relate noisy QNN with the differentially private (DP) learning model \cite{dwork2006calibrating}, driven by the observation that  DP models share a similar behavior with noisy QNN. Specifically, analogous to QNN, DP models involves certain types of noise to achieve the privacy guarantee. If noisy QNN were also a DP model, then we can conclude the same learnability of QNN and noisy QNN, since  a concept class that is learnable by a (non-private) algorithm with polynomial sample complexity can also be learned privately using a polynomial number of samples \cite{kasiviswanathan2011can}. The Lemma \ref{lem:DP_QNN01} provides an affirmative response, which exhibits that QNN with noisy gates can be treated as a DP learning model (The proof details is given in Appendix \ref{Appendix:subsec_DP_QNN_secv}).

The learnability of DP models \cite{kasiviswanathan2011can} has been extensively explored  in the literature. Two studies \cite{arunachalam2017guest,arunachalam2020quantum} separately proved that the sample complexities of classical and quantum (differentially private) $\PAC$ learning are equal, up to constant factors. Combining with the fact that $\PAC = \PPAC$ \cite{kasiviswanathan2011can}, we can conclude that the sample complexity of $\PAC$, $\PPAC$, quantum $\PAC$, and quantum $\PPAC$ learning are equivalent.  The conclusion together with  Lemma \ref{lem:DP_QNN01}  allows us to achieve Theorem \ref{thm:DP_QNNQAE_inform}, i.e., if noiseless QNN $\PAC$ learns a concept class, then QNN with gate errors can also learn this    concept class using polynomial number of samples. We present the full  proof of Theorem \ref{thm:DP_QNNQAE_inform} in Appendix \ref{Appendix:proof_thm2_QNN_learnability}.

The result of Theorem \ref{thm:DP_QNNQAE_inform} indicates that there does not exist a  concept class that can be efficiently learned by noiseless QNN, while it is computationally hard for noisy QNN. This result provides a theoretical guarantee to realize  QNN on NISQ chips to seek potential quantum advantages.

We further utilize the theoretical results of $\QSQ$ model to quantify what kinds of learning problems can be tackled by noisy QNN with quantum advantages. The study \cite{arunachalam2020quantum} shows that  $\QSQ$ can efficiently tackle parity, juntas, and  DNF learning tasks, which are provably hard to learn by the classical statistical query $(\SQ)$ models \cite{blum2003noise}. We proved in Theorem \ref{thm:QNN_QAE_QSQ_info} that the $\QSQ$ model can be efficiently simulated by QNN, and whose proof is given in Appendix \ref{Appendix:proof_thm3_QNN_learnability}. Therefore we conclude that these tasks can also be accomplished by QNN with quantum advantage.

All results in this section assume depolarization gate noise; however, they can be extended to a more general model of gate noise. See Appendix \ref{Appendix:sec_more_general_channel}  for details.

\section{Numerical simulations}\label{appen:numerical_sim}

 We employ the UCI ML hand-written digits datasets \cite{Dua_2019} to validate the correctness of utility bounds $R_1$ and $R_2$ of QNN, as achieved in Section \ref{sec:learn_QNN} and \ref{sec:learnability_QNN_QAE}. In the rest of this section, we first introduce the employed dataset and the required preprocessing steps. We then elaborate the employed parameterized quantum circuits that are used in QNN. We last demonstrate our numerical simulation results. 
 
 The employed dataset includes in total $1797$ hand-written digits images with  $10$ labels, where each label refers to a digit and each image has $64$ attributes. The data preprocessing has three steps. First, we clean the dataset and only collect images with labels $0$ and $1$. After cleaning, the total number of images is $360$, where the  number of examples with label $0$ (label $1$) is $178$ ($172$).  Some collected examples are shown in the left panel of Fig.~\ref{fig:digit_PQC}. Alternatively, our simulation focuses on the binary classification task. Second, we utilize a feature dimension reduction technique, i.e., principal component analysis (PCA) \cite{wold1987principal}, to reduce the feature dimension of each data example from $64$ to $3$. The middle panel of Fig.~\ref{fig:digit_PQC} exhibits the reconstructed hand-written digit images using the reduced data features. Such a step aims to balance the relatively high dimension features of the data example and the limited quantum resources available in present-day. After applying PCA, we denote the employed dataset as $\bm{z}=\{(\bm{x}_i, y_i)\}_{i=1}^{360}$, where $\bm{x}_i\in\mathbb{R}^3$ is the $i$-th data feature and ${y}_i\in\{0,1\}$ is the $i$-th label. The last step is uniformly  and randomly   splitting the dataset $\bm{z}$ into two groups, i.e., the training dataset $\bm{z}_t$ and the test dataset $\bm{z}_p$. The size of the training dataset $\bm{z}_t$ and the test dataset $\bm{z}_p$ is $280$ and $80$, respectively.     
 
  \begin{figure*}[h] 
	\centering
\includegraphics[width=0.84\textwidth]{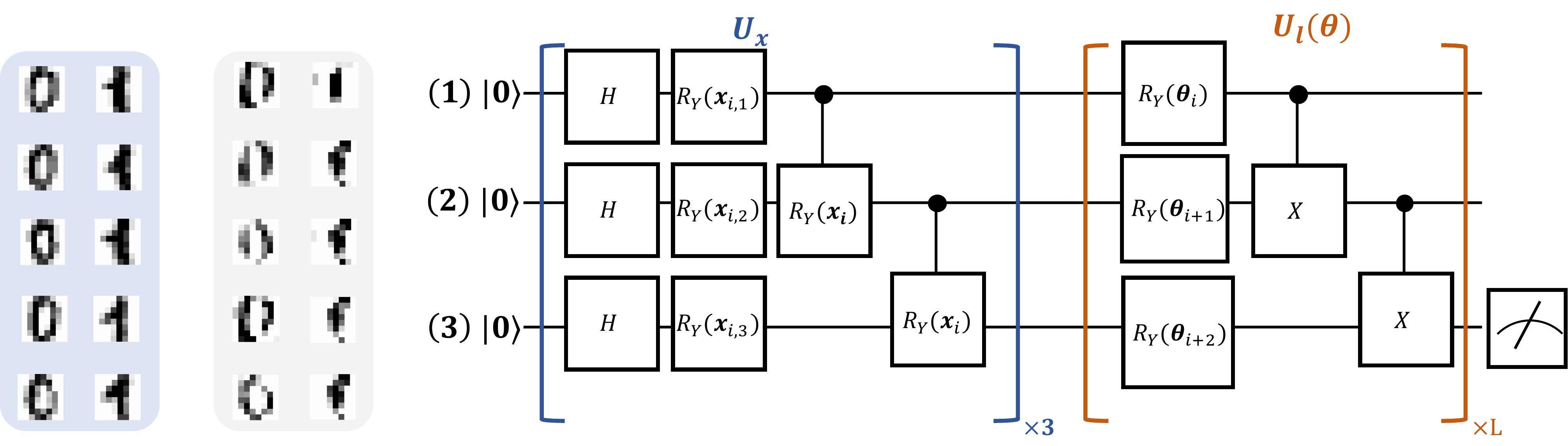}
\caption{\small{The training examples of hand-written digits and the implementation of quantum circuits. The left panel and middle panel illustrate the original and reconstructed training examples, respectively. The right panel demonstrates the implementation of data encoding circuit and trainable circuit used in QNN. The label `x$3$' and `x$L$' means repeating  the quantum gates in blue and brown boxes with $3$ and $L$ times, respectively.}}
\label{fig:digit_PQC}
\end{figure*}

 The construction of parameterized quantum circuit, i.e., the data encoding circuit $U_{\bm{x}}$ and the trainable unitary $U(\bm{\theta})$, follows the proposal \cite{havlivcek2019supervised}. In particular, the data encoding circuit $U_{\bm{x}}$ uses the kernel encoding method, and the architecture of trainable unitary $U(\bm{\theta})$ follows the layer structure. The right panel of Fig.~\ref{fig:digit_PQC} illustrates the implementation of data encoding circuit and trainable circuit used in QNN. Three qubits are employed to build such two circuits. The data encoding circuits $U_{\bm{x}}$ is composed of Hadamard gates $H=\frac{1}{\sqrt{2}}\big(\begin{smallmatrix}
  1 & 1 \\
  1 & -1
\end{smallmatrix}\big)$, $R_Y$ gates  $R_Y(2a)= \big(\begin{smallmatrix}
  \cos(a) & -\sin(a) \\
  \sin(a) & \cos(a)
\end{smallmatrix}\big)$, and controlled-$R_Y$ gates $\text{CRY}(2a)=\ket{0}\bra{0}\otimes \mathbb{I}_2 + \ket{1}\bra{1}\otimes R_Y(2a)$. Specifically,  the rotation angle in $R_Y(\bm{x})$ is $(\pi-\bm{x}_{i,1})(\pi-\bm{x}_{i,2})(\pi-\bm{x}_{i,3})$.  The construction of trainable circuits $U(\bm{\theta})$ uses $R_Y$ gates and controlled-NOT gates $CX=\ket{0}\bra{0}\otimes \mathbb{I}_2 + \ket{1}\bra{1}\otimes X$ with $X=\big(\begin{smallmatrix}
  0 & 1 \\
  1 & 0
\end{smallmatrix}\big)$.

 \begin{figure*}[h]
	\centering
\includegraphics[width=0.98\textwidth]{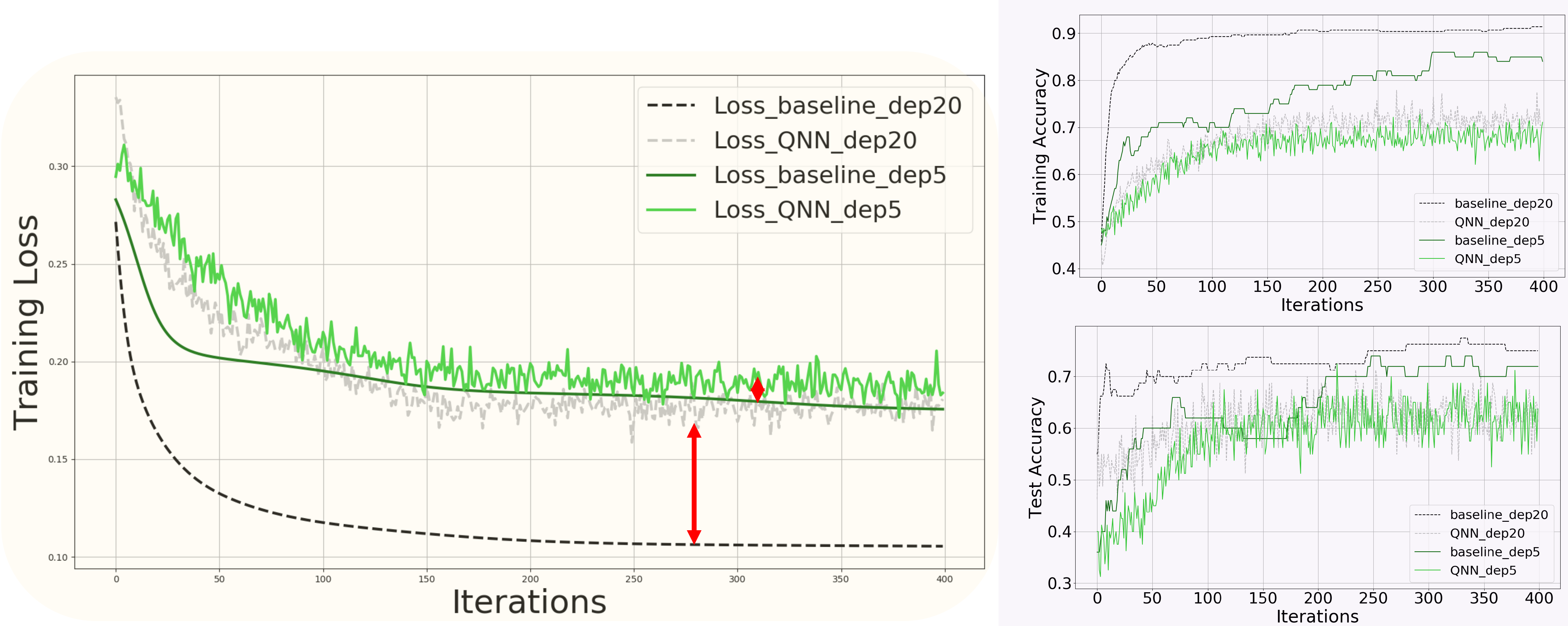} 
\caption{\small{The simulation results of QNN on  hand-written digit  dataset. The left panel shows the training loss under different hyper-parameters settings. In particular, the label `Loss\_baseline\_dep20' (`Loss\_baseline\_dep5') refers to the obtained loss under the  setting $L=20$ ($L=5$),  $p=0$, and $K\rightarrow\infty$, where $L$, $p$, and $K$ refer to the  circuit depth, depolarization rate, the number of measurements to estimate expectation value used in QNN,  respectively. Similarly, the label `Loss\_QNN\_dep20' (`Loss\_QNN\_dep5') refers to the obtained loss of QNN under the setting $L=20$ ($L=5$),  $p=0.0025$, $K=20$.  The upper right and lower right panels separately demonstrate the training accuracy and test accuracy of the quantum classifiers with different hyper-parameters settings. }}
\label{fig:digit_sim}
\end{figure*}

We now employ the preprocessed hand-written digits dataset and quantum circuits as described above to study the learnability of QNN under depolarization noise. Specifically, we apply depolarization channel $\mathcal{N}_{p}$ to every  quantum circuit depth, where the depolarization rate is set as $p=0.0025$. The depth of trainable circuits $U(\bm{\theta})$ is set as $L=5$ and $L=20$, respectively. The corresponding number of trainable parameters is $15$ and $60$, respectively. We also train QNN without noisy channels $\mathcal{N}_{p}$ under the setting $L=5,20$, which aims to estimate the optimal parameter $\bm{\theta}^*$ and the minimized objective function $\mathcal{L}^*$. The number of iterations for all numerical simulations is set as $T=400$. For QNN, the number of measurements to estimate the expectation value is set as $K=20$. 

The simulation results are shown in Fig.~\ref{fig:digit_sim}. We now elaborate how numerical simulations accord with our theoretical results. Two red arrows indicate the gap between optimal result $\mathcal{L}^*$ and the achieved results $\mathcal{L}(\bm{\theta}^{(T)})$. With increasing the circuit depth $L$, the gap becomes large for  QNN. Such a phenomenon follows our theoretical result, where a larger $d$ and $\tilde{p}$ lead a poor utility bound $R_2$.

\section{Conclusion}
In this study, we explore the learnability of QNN from the aspect of ERM framework and sample complexity.  The achieved utility bounds towards ERM  indicate that, more measurements, lower noise, and shallower circuit depth contribute to a better performance of QNN. Built on the conclusion that the same  learnability between the noiseless QNN and QNN with noisy gates, we obtain the  theoretical evidence that supports implementation of QNN on NISQ chips to pursue quantum advantages. Moreover, we demonstrate that QNN with noisy gates can efficiently learn parity, juntas, and DNF with quantum advantages even with gate noise.    

Our work also generates plausible new directions for NISQ study that we plan to explore in the future. First, we will use other advanced DP results to analyze various variational hybrid models on NISQ machines with provable guarantees. Second, we aim to tackle private learning tasks with quantum merits because the gate noise of NISQ machines will benefit the design of quantum DP mechanism.

 \newpage
 
\bibliographystyle{unsrt}

\newpage
\appendix

The organization of the appendix is as follows. In Appendix  \ref{appd:Qc_pre},  we unify the notations used in the whole appendix. In Appendix \ref{appen:QNN},  we elaborate the implementation details of the quantum encoding circuit $U_{\bm{x}}$ and the trainable quantum circuit $U(\bm{\theta})$ used in QNN. In Appendix \ref{Append:obj_func_prop}, we present the proof of Lemma  \ref{lem:Lsmmoth}, which quantifies the properties of the objective function with respect to the optimization theory. Then, in Appendix \ref{appendix:subsec:QNN_grad_dist}, we exhibit the proof of Theorem \ref{thm:noise_QNN_gaussian}, as the precondition to achieve utility bounds of QNN.  In  Appendix \ref{Appendix:Thm_utl_QNN}, we exhibit the proofs details of Theorem \ref{thm:informal_utl_QNNQAE_DP} that achieves the utility bounds of QNN towards ERM. The following four  sections explain the learnability of QNN from the perspective of sample complexity. Specifically, in Appendix \ref{Appendix:subsec_DP_QNN_secv}, we provide the proof details of Lemma \ref{lem:DP_QNN01}.  Next, in Appendix \ref{Appendix:proof_thm2_QNN_learnability} and \ref{Appendix:proof_thm3_QNN_learnability}, we separately prove Theorem \ref{thm:DP_QNNQAE_inform} and Theorem \ref{thm:QNN_QAE_QSQ_info}. Eventually, in Appendix \ref{Appendix:sec_more_general_channel}, we generalize all achieved results to a more general quantum channel.

\section{The summary of notations}\label{appd:Qc_pre}

Here we unify the notations used in the appendices.   A random variable $X$ that follows Delta distribution is denoted as $X\sim \Del(x_0)$, i.e., $\Pr(X=x_0)=1$ and $\Pr(X\neq x_0)=0$.  A random variable $X$ that follows uniform distribution is denoted as $X\sim U(a, b)$, where $P(X=x_0)=1/(b-a)$ with $a \leq x_0\leq b$.  We denote the $\ell_p$ norm of $\mathbf{v}$ as $\|{\mathbf{v}}\|_p$. In particular,  $\|{\mathbf{v}}\|$ refers to the $\ell_2$ norm.

\section{Implementation details of encoding circuit and trainable  circuit of QNN}\label{appen:QNN}
 
The selection of encoding circuits $U_{\bm{x}}$ and trainable circuit $U(\bm{\theta})$ is flexible in QNN. We now separately explain the implementation details of these two circuits supported by QNN.  
  
\textbf{Encoding circuit $U_{\bm{x}}$.} The typical encoding circuits can be divided into four categories. A common feature of these encoding methods is that their implementation only costs low circuit depth, driven by the restricted quantum resources. The first category is the direct amplitude encoding \cite{du2018implementable,plesch2011quantum,schuld2017implementing,schuld2020circuit}. Specifically, the encoder circuit  satisfies $U_{\bm{x}}: \mathcal{B}_i \rightarrow \frac{1}{\sqrt{B_s}}\sum_{b=1}^{B_s}\sum_{j=1}^D\hat{\bm{x}}^{(i)}_{b,j}\ket{b}\ket{j}$ with $\hat{\bm{x}}^{(i)}_{b,j} = {\bm{x}}^{(i)}_{b,j}/\|\bm{x}^{(i)}_{b,j}\|$. This method  requires a low feature dimension   $D$, since the quantum gates complexity to build $U_{\bm{x}}$ is $O(D)$.  The second category is the kernel mapping \cite{havlivcek2019supervised,mitarai2018quantum,schuld2019quantum}, where $\mathcal{B}_i$ is encoded into a set of  single-qubit gates with a specified arrangements, e.g.,  $U_{\bm{x}}(\mathcal{B}_i)=\sum_{b=1}^{B_s}(\ket{b}\bra{b})\otimes_{j=1}^D \RY(\bm{x}^{(i)}_{b,j}) $. The third category is the dimension reduction method proposed by \cite{wilson2018quantum}. Specifically, instead of encoding $\mathcal{B}_i$, the  amplitude or kernel encoder circuits $U_{\bm{x}}$ is exploited to encode a projected  features $g(\mathcal{B}_i)\in \mathbb{R}^{B_s\times D'}$, where $g(\cdot)$ is a predefined function and $D'\ll D$. The fourth category is the basis encoding \cite{arunachalam2017guest,arunachalam2020quantum,farhi2018classification}, which is broadly used in quantum learning theory. Specifically, the encoding circuit $U_{\bm{x}}$ is employed to prepare a quantum example $\ket{\psi}=\sum_{\bm{x}\in\{0,1\}^N}\sqrt{\mathcal{D}(\bm{x})}\ket{\bm{x},c(\bm{x})}$, where $\mathcal{D}(\bm{x})$ is the data distribution over $\bm{x}$, $c(\bm{x})$ corresponds to the label of the bit-string $\bm{x}$ \cite{arunachalam2017guest,atici2005improved}. In most cases, the distribution $\mathcal{D}(\bm{x})$ is uniform. Hence, the state $\ket{\psi}$ can be efficiently  prepared by  setting $B=1$, and applying Hadamard gates and control-not gates \cite{nielsen2010quantum} to the initial state $\ket{0}^{\otimes N+1}$.

\textbf{Trainable quantum circuits $U(\bm{\theta})$.} The trainable  quantum circuits, a.k.a, parameterized quantum circuits  \cite{benedetti2019parameterized,du2018expressive}, used in QNN can be written as a product of layers of unitaries in the form $U(\bm{\theta})=\prod_{l=1}^L U_l(\bm{\theta}_l)$, where  $U_l(\bm{\theta}_l)$ is composed of parameterized single-qubit gates and fixed two-qubits gates.  Each trainable layer can be decomposed into $U_l(\bm{\theta}_l)=(\bigotimes_{k=1}^N U_{l,k}(\bm{\theta}_l))U_{eng}$, where $U_{l,k}(\bm{\theta}_l)$ represents the composition of trainable single-qubit gates and $U_{eng}$ refers to entanglement layer that contains two-qubits gates.  Depending on the detailed architecture, the implementation of $U_l(\bm{\theta}_l)$ can be categorized into three classes. The first class is the hardware-efficient circuit architecture, where the selection of $U_k(\bm{\theta}_l))$ and $U_{eng}$  is  according to the given NISQ machine that  has the specific sparse qubit-to-qubit connectivity and a specified set of  quantum gates \cite{cerezo2020cost,kandala2017hardware,mcclean2018barren}.  The second class is the tensor network inspired architecture. In particular, the layout of quantum gates is following different tensor networks, e.g., the matrix product state, the tree tensor network, and the multi-scale entanglement renormalization ansatz (MERA) \cite{huggins2019towards}. The third class is the Hamiltonian based architecture, where the entanglement layer $U_{eng}$  refers to a specific Hamiltonian, e.g.,  the study \cite{mitarai2018quantum} employs $U_{eng}=e^{-iHT}$ with $H=\sum_{j=1}^N a_j\X_j +\sum_{j=1}^N\sum_{k=1}^{j-1}J_{jk}\Z_i\Z_k$. Notably, almost all quantum approximate optimization algorithms follow the Hamiltonian based architecture \cite{farhi2014quantum}.

\section{Proof of Lemma \ref{lem:Lsmmoth}}\label{Append:obj_func_prop}
 
 The Lemma \ref{lem:Lsmmoth} indicates that the objective function $\mathcal{L}(\bm{\theta})$ used in QNN satisfies $S$-smooth and $G$-Lipschitz properties. Moreover, it also satisfies the Polyak-Lojasiewicz (PL) condition under the assumption with $\lambda\geq 1/\pi$. To ease the discussion, we first formulate the explicit form of $\mathcal{L}(\bm{\theta})$. Without loss of generality, we set $B=n$, where each batch $\mathcal{B}_i$ only contains the $i$-th input $\bm{x}_i$. Denote the prepared quantum states as $\{\rho_{\mathcal{B}_i}\}_{i=1}^n$ i.e., $\rho_{\mathcal{B}_i}=\ket{\phi_{\mathcal{B}_i}}\bra{\phi_{\mathcal{B}_i}}$ and $\ket{\phi_{\mathcal{B}_i}}\xleftarrow{U_{\bm{x}}} \{\bm{x}_i\}$ refers to the quantum example corresponding to the classical input batch $\mathcal{B}_i$ (or equivalently, $\bm{x}_i$).  The explicit form of the objective function is  
 \begin{align}\label{eqn:append_lem3_0}
 \mathcal{L}(\bm{\theta}, \bm{z}) = \frac{1}{n}\sum_{i=1}^n  \left(  \hat{y}_i - y_i  \right)^2 + \frac{\lambda}{2}\|\bm{\theta}\|_2^2~,
 \end{align}
 where  $\hat{y}_i= \Tr(\Pi U(\bm{\theta})\rho_{\mathcal{B}_i}U(\bm{\theta})^{\dagger})$  refers to the prediction of QNN given the $i$-th input $\bm{x}_i$, $U(\bm{\theta})$ is the trainable circuit, $\Pi$ is the employed two-outcome POVM, and $y_i$ is the true label of the $i$-th input. Moreover, since the tunable parameters $\bm{\theta}$ in QNN refer to the rotation angles, we set its range as $\bm{\theta}\in[\pi, 3\pi]^{d}$.  
 
 \begin{proof}[Proof of Lemma \ref{lem:Lsmmoth}]
   We employ the  three lemmas presented below to prove Lemma \ref{lem:Lsmmoth}, whose proofs are given in the following subsections. 

 \begin{lem} \label{lem:Lsmooth}
 The objective function $\mathcal{L}$ is $S$-smooth with $S = (3/2+ \lambda)$.
 \end{lem} 
 \begin{lem}\label{lem:G_lipschitz}
 The objective function $\mathcal{L}$ is $G$-Lipschitz with $G=d(1+3\pi\lambda)$.  
 \end{lem}
  \begin{lem}\label{lem:PL_QNN_QAE}
Assume $\lambda > \frac{1}{\pi}$.  The objective function $\mathcal{L}$  satisfies  PL condition with $\mu = \frac{(-1+ \lambda \pi)^2}{1+{\lambda d}(3\pi)^2}$.
 \end{lem}

 	In conjunction with the results of Lemma \ref{lem:Lsmooth}, \ref{lem:G_lipschitz}, and \ref{lem:PL_QNN_QAE}, the proof of Lemma \ref{lem:Lsmmoth} is completed.  
 \end{proof} 

\subsection{Proof of $S$-smooth,  Lemma \ref{lem:Lsmooth}}\label{Sec:app_smooth}
 \begin{proof}[Proof of Lemma \ref{lem:Lsmooth}]

  Recall the function $\mathcal{L}(\bm{\theta})$ is $S$-smooth if 
  \begin{equation}\label{eqn:lem_smmth_h1_0}
  	  \nabla^2 \mathcal{L}(\bm{\theta})\preceq S\mathbb{I}~,
  \end{equation}
 with $S>0$. In other words, to obtain $S$, we need to obtain the upper bound of the second derivative of $\mathcal{L}(\bm{\theta})$, i.e., $S\geq \| \nabla^2 \mathcal{L}(\bm{\theta})\|_{\infty}$.
  
 Following the notation used in Eqn.~(\ref{eqn:append_lem3_0}), the gradient  for the parameter $\bm{\theta}_j$ is     
  \begingroup
  \allowdisplaybreaks
  \begin{align}\label{eqn:lem_smmth_h1_1}
   	&\frac{\partial \mathcal{L}(\bm{\theta})}{\partial \bm{\theta}_j} \nonumber\\
   	= &  \frac{2}{n}\sum_{i=1}^n \left(\hat{y}_i -y_i\right)\frac{\partial \hat{y}_i}{\partial \bm{\theta}_j} + \frac{\lambda}{2}\frac{\partial \|\bm{\theta}\|^2_2}{\partial \bm{\theta}_j} \nonumber \\
   	= & \frac{2}{n}\sum_{i=1}^n \left(\hat{y}_i -y_i\right) \frac{\hat{y}_{i}^{(+_j)}-\hat{y}_{i}^{(-_j)}}{2} + \lambda\bm{\theta}_j \nonumber \\
\leq & 1 + 3 \lambda \pi ~,
   \end{align}
   \endgroup
   where $\hat{y}_{i}^{(\pm_j)}=\Tr(\Pi U(\bm{\theta} \pm\frac{\pi}{2}\bm{e}_j) \rho_{\mathcal{B}_i}U(\bm{\theta} \pm\frac{\pi}{2}\bm{e}_j)^{\dagger} )$,  the second equality employs the conclusion of the parameter shift rule with $\frac{\partial \hat{y}_i}{\partial \bm{\theta}_j}= \frac{\hat{y}_{i}^{(+_j)}-\hat{y}_{i}^{(-_j)}}{2}$ \cite{mitarai2018quantum, schuld2019evaluating}, and the last inequality uses the facts $\pi \leq \bm{\theta}_j\leq 3\pi$, $(\hat{y}_i -y_i)\leq 1$, and $\hat{y}_{i}^{(+_j)}-\hat{y}_{i}^{(-_j)}\leq 1$, since $\hat{y}_i, y_i, \hat{y}_{i}^{(\pm_j)}\in [0, 1]$. 

The upper bound of the derivative $\frac{\partial^2 \mathcal{L}(\bm{\theta})}{\partial \bm{\theta}_j\partial \bm{\theta}_k}$ can be derived using the results of Eqn.~(\ref{eqn:lem_smmth_h1_1}). In particular, 
\begin{align}\label{eqn:lem_smmth_h1_2}
	& \frac{\partial^2 \mathcal{L}(\bm{\theta})}{\partial \bm{\theta}_j\partial \bm{\theta}_k} = \frac{\partial(\frac{\partial\mathcal{L}(\bm{\theta})}{\partial \bm{\theta}_j})  }{\partial \bm{\theta}_k}  = \frac{1}{n}\sum_{i=1}^n  \frac{\partial \left ( \left(\hat{y}_i -y_i\right)  \left( \hat{y}_{i}^{(+_j)}-\hat{y}_{i}^{(-_j)}\right)  + \lambda\bm{\theta}_j \right) }{\partial \bm{\theta}_k} \nonumber\\
	= & \frac{1}{n}\sum_{i=1}^n  \left[\frac{\partial \hat{y}_i}{\partial \bm{\theta}_k} \left( \hat{y}_{i}^{(+_j)}-\hat{y}_{i}^{(-_j)}\right) + \left(\hat{y}_i -y_i\right) \frac{\partial \left( \hat{y}_{i}^{(+_j)}-\hat{y}_{i}^{(-_j)}\right)}{\partial \bm{\theta}_k} + \lambda  \right]    \nonumber\\ 
	 \leq & \frac{3}{2} + \lambda~,
\end{align}
where the first equality comes from the last equality of Eqn.~(\ref{eqn:lem_smmth_h1_1}), and  the last inequality employs $(\hat{y}_i -y_i)\leq 1$, $\hat{y}_{i}^{(+_j)}-\hat{y}_{i}^{(-_j)}\leq 1$, and \[ \frac{\partial  \hat{y}_{i} }{\partial \bm{\theta}_k}, \frac{\partial   \hat{y}_{i}^{(+_j)} }{\partial \bm{\theta}_k}, \frac{\partial   \hat{y}_{i}^{(-_j)} }{\partial \bm{\theta}_k} \in[-1/2,1/2]~,\] 
supported by the parameter shit rule and $\hat{y}_i, \hat{y}_{i}^{(\pm_j)}\in [0, 1]$. 

The result of Enq.~(\ref{eqn:lem_smmth_h1_2}) implies that $\|\nabla^2 \mathcal{L} \|_{\infty}\leq \frac{3}{2} + \lambda$. In conjunction with Eqn.~(\ref{eqn:lem_smmth_h1_0}), the objective function is $S$-smooth with $S=\frac{3}{2} + \lambda$. 
\end{proof}

\subsection{Proof of $G$-Lipschitz, Lemma \ref{lem:G_lipschitz}}\label{Sec:app_Liptz}
\begin{proof}[Proof of Lemma \ref{lem:G_lipschitz}]

 Recall a function $f(\bm{x})$ is  $G$-Lipschitz if it satisfies  
 \begin{equation}\label{eqn:lem_G_1}
 	|f(\bm{b})-f(\bm{a})| \leq G\|\bm{b}-\bm{a}\|~.
 \end{equation}
 
Moreover, the mean value theorem gives that, if 
 $f: \mathbb{R}^d\rightarrow \mathbb{R}$ is differentiable and $[\bm{a},\bm{b}]\subseteq  \mathbb{R}^d$, then $\exists \bm{c}\in(\bm{a}, \bm{b})$ such that 
  \begin{equation}\label{eqn:mvt}
 	f(\bm{b})-f(\bm{a}) = \langle \nabla f(\bm{c}), \bm{b}-\bm{a} \rangle~.
 \end{equation}

 Combining Enq.~(\ref{eqn:lem_G_1}) and (\ref{eqn:mvt}), the $G$-Lipschitz condition in Eqn.~(\ref{eqn:lem_G_1}) is equivalent to   
 \begin{equation}\label{eqn:lem_G_2}
 	|\langle \nabla f(\bm{c}), \bm{b}-\bm{a} \rangle |\leq G
 	\|\bm{b}-\bm{a}\|~.
 \end{equation}

We now replace $f$, $\bm{b}$, and $\bm{a}$ used in Eqn.~(\ref{eqn:lem_G_2})  with $\mathcal{L}$, $\bm{\theta}^{(1)}$, and $\bm{\theta}^{(2)}$ to prove that the objective function $\mathcal{L}$ is $G$-Lipschitz. Specifically, we need to find a real value $G$ that satisfies 
\begin{equation}\label{eqn:lem_G_3}
	\left| \left\langle \nabla \mathcal{L}(\bm{\theta}) , \bm{\theta}^{(1)} -  \bm{\theta}^{(2)} \right  \rangle \right| \leq G\|\bm{\theta}^{(1)} -  \bm{\theta}^{(2)}\|~,
\end{equation}
where $\bm{\theta}\in (\bm{\theta}^{(2)}, \bm{\theta}^{(1)})$. 

The upper bound of the term $\left\langle \nabla \mathcal{L}(\bm{\theta}), \bm{\theta}^{(1)} -  \bm{\theta}^{(2)} \right  \rangle $ is 
\begin{align}\label{eqn:lem_G_4}
	 \left\langle \nabla \mathcal{L}(\bm{\theta}), \bm{\theta}^{(1)} -  \bm{\theta}^{(2)} \right  \rangle  
	\leq   \left\| \nabla \mathcal{L}(\bm{\theta}) \right\| \|\bm{\theta}^{(1)} -  \bm{\theta}^{(2)} \| \leq d \left\|\nabla \mathcal{L}(\bm{\theta}) \right\|_{\infty}  \|\bm{\theta}^{(1)} -  \bm{\theta}^{(2)} \|~.
\end{align}

In conjunction with Eqn.~(\ref{eqn:lem_G_3}) and (\ref{eqn:lem_G_4}), $G$-Lipschitz of $\mathcal{L}$ requests 
\begin{equation}\label{eqn:lem_G_5}
	d \left\|\nabla \mathcal{L}(\bm{\theta}) \right\|_{\infty} \leq G~.
\end{equation}

By leveraging the result of Eqn.~(\ref{eqn:lem_smmth_h1_1}) with $\nabla_j \mathcal{L}(\bm{\theta}) \leq 1+ 3\lambda \pi$, we obtain the upper bound of the left side in Eqn.~(\ref{eqn:lem_G_5}) is 
\begin{equation}
		d \left\|\nabla \mathcal{L}(\bm{\theta}) \right\|_{\infty} \leq d (1 + 3\pi \lambda)~.
\end{equation}
This leads to the objective function $\mathcal{L}$ of QNN satisfying $G$-Lipschitz with $G=d(1+3\pi\lambda)$.
\end{proof}
 
 \subsection{Proof of PL condition, Lemma  \ref{lem:PL_QNN_QAE}}
\begin{proof}[Proof of Lemma  \ref{lem:PL_QNN_QAE}]
Recall the definition of Polyak-Lojasiewicz as formulated in Definition \ref{def:S-smoo-G_lip}, it requires that the objective function $\mathcal{L}$ satisfies 
\begin{equation}\label{eqn:PL_0}
	\|\nabla \mathcal{L}(\bm{\theta}) \|^2\geq 2\mu(\mathcal{L}(\bm{\theta})-\mathcal{L}^*)~,
\end{equation}
where $\mathcal{L}^* = \min_{\bm{\theta}\in\mathcal{C}} \mathcal{L}(\bm{\theta})$.
 
We first derive a lower bound of $\|\nabla \mathcal{L}(\bm{\theta}) \|^2$. In particular, we have
\begin{align}\label{eqn:PL_1}
	& \|\nabla \mathcal{L}(\bm{\theta}) \|^2 = \sum_{j=1}^d  (\nabla_j \mathcal{L}(\bm{\theta}_j))^2
	 \geq   \max_j  (\nabla_j \mathcal{L}(\bm{\theta}))^2  ~. 
\end{align}
  
The lower bound of $ \max_j  (\nabla_j \mathcal{L}(\bm{\theta}))^2 $ as shown in  Eqn.~(\ref{eqn:PL_1}) follows  
\begin{equation}\label{eqn:PL_5}
	\max_j  (\nabla_j \mathcal{L}(\bm{\theta}))^2   
	 \geq  (-1+ \lambda \pi)^2~,
\end{equation}
where the last inequality is achieved by exploiting the last second line of Eqn.~(\ref{eqn:lem_smmth_h1_1}), and the facts $\bm{\theta}_j\in[\pi,3\pi]$ and $\hat{y}_i, y_i, \hat{y}_{i}^{(\pm_j)}\in [0, 1]$, i.e.,  \[\nabla_j \mathcal{L}(\bm{\theta}) = \frac{2}{n}\sum_{i=1}^n \left(\hat{y}_i -y_i\right) \frac{\hat{y}_{i}^{(+_j)}-\hat{y}_{i}^{(-_j)}}{2} + \lambda\bm{\theta}_j \geq -1 + \lambda \pi~.\]

Combining the assumption $\lambda\geq 1/\pi$ and the above results, the lower bound of Eqn.~(\ref{eqn:PL_1})  satisfies  \[\|\nabla \mathcal{L}(\bm{\theta}) \|^2 \geq (-1+ \lambda \pi)^2>0~.\]

We then derive the upper bound of the term $(\mathcal{L}(\bm{\theta})-\mathcal{L}^*)$ in Eqn.~(\ref{eqn:PL_0}). In particular, we have
\begin{align}\label{eqn:PL_2}
	& \mathcal{L}(\bm{\theta})-\mathcal{L}^*  
	\leq  \mathcal{L}(\bm{\theta})  + 0 \leq 1 +{\lambda d}(3\pi)^2~,  
\end{align}
where the first inequality comes from the definitions of $\mathcal{L}^*$, i.e., \[-\mathcal{L}^*= - \frac{1}{n} \sum_{i=1}^n (\hat{y}_i^* - y_i)^2 -  \frac{\lambda}{2}\|\bm{\theta}\|^2 \leq 0~,\] with $\hat{y}_i^* = \Tr(\Pi U(\bm{\theta}^*)\rho_{i}U(\bm{\theta}^*)^{\dagger})$, and the second inequality employs the definition of $\mathcal{L}(\bm{\theta})$ with  \[\mathcal{L}(\bm{\theta})= \frac{1}{n} \sum_{i=1}^n (\hat{y}_i - y_i)^2 +  \frac{\lambda}{2}\|\bm{\theta}\|^2 \leq 1 + \frac{\lambda}{2}\|\bm{\theta}\|^2~,\] and $\frac{\lambda}{2}\|\bm{\theta}\|^2\leq \frac{\lambda}{2}d\|\bm{\theta}\|_{\infty}^2=(3\pi)^2 \lambda d/2$.

By combining  Eqn.~(\ref{eqn:PL_5}) and (\ref{eqn:PL_2}) with  Eqn.~(\ref{eqn:PL_0}), we obtain the following relation 
\begin{equation}
\|\nabla \mathcal{L}(\bm{\theta}) \|^2\geq 	(-1+ \lambda \pi)^2 \geq 2\mu  (1+{\lambda d}(3\pi)^2) \geq 2\mu(\mathcal{L}(\bm{\theta})-\mathcal{L}^*) ~.
\end{equation}
The above relation indicates that the objection function $\mathcal{L}(\bm{\theta})$ satisfies PL  condition with
\[\mu =  \frac{(-1+ \lambda \pi)^2}{1 +{\lambda d}(3\pi)^2}~. \]

\end{proof}

\section{Proof of Theorem \ref{thm:noise_QNN_gaussian}}\label{appendix:subsec:QNN_grad_dist}
Theorem \ref{thm:noise_QNN_gaussian} establishes the relation between the analytic gradient $\nabla_j {\mathcal{L}}_i(\bm{\theta}^{(t)})$ and the estimated gradient $\nabla_j\bar{\mathcal{L}}_i(\bm{\theta}^{(t)})$ of QNN. Its formal description  is as follows.  

\begin{thm}[The formal  description of Theorem \ref{thm:noise_QNN_gaussian}]\label{thm:noise_QNN_gaussian_formal} 
Denote $\tilde{p}=1-(1-p)^{L_Q}$ with $L_Q$ being the quantum circuit depth. At the $t$-th iteration, we  define five constants with
 \[
    C^{(i,t)}_{j,a}= 
\begin{cases}
    (1-\tilde{p})\tilde{p}(1/2-{Y}_i)(\hat{Y}_i^{(t,+_j)}-\hat{Y}_i^{(t,-_j)}) - (2\tilde{p}-\tilde{p}^2)\lambda\bm{\theta}_j^{(t)}~,&  a=1\\
    (1-\tilde{p})(\hat{Y}_i^{(t,+_j)}-\hat{Y}_i^{(t,-_j)})~,  & a=2 \\
 ((1-\tilde{p})\hat{Y}_i^{(t)} + \tilde{p}/2-{Y}_i)~,   & a=3  \\
 \frac{-(1-\tilde{p})(\hat{Y}_i^{(t)})^2 + (1-\tilde{p})^2\hat{Y}_i^{(t)}+\frac{\tilde{p}}{2}-\frac{\tilde{p}^2}{4}}{K}~, & a=4 \\
   \frac{-(1-\tilde{p})((\hat{Y}_i^{(t,+_j)})^2+(\hat{Y}_i^{(t,-_j)})^2) + (1-\tilde{p})^2(\hat{Y}_i^{(t, +_j)}+\hat{Y}_i^{(t, -_j)})+\tilde{p}-\frac{\tilde{p}^2}{2}}{K}~, & a =5~,
\end{cases}
\]
where $\hat{Y}_i^{(t,\pm_j)}= \Tr(\Pi U(\bm{\theta}\pm \bm{e}_j)\rho_{\mathcal{B}_i}U(\bm{\theta}\pm \bm{e}_j)^{\dagger})$, $K$ refers to the number of quantum  measurements, and $\hat{Y}_i^{(t)}$ and $Y_i$ are  the sum average of the predicted  and true labels for the $i$-th batch $\mathcal{B}_i$.  

 The relation between the estimated and analytic gradients follows  \[\nabla_j\bar{\mathcal{L}}_i(\bm{\theta}^{(t)}) = (1-\tilde{p})^2\nabla_j{\mathcal{L}}_i(\bm{\theta}^{(t)}) + C_{j,1}^{(i,t)} +  \bm{\varsigma}_i^{(t,j)}\] with $\bm{\varsigma}_i^{(t,j)}= C_{j,2}^{(i,t)}\xi_i^{(t)} + C_{j,3}^{(i,t)}\xi_i^{(t,j)}+\xi^{(t)}\xi_i^{(t,j)}$, where $\xi_i^{(t)}$ and $\xi_i^{(t,j)}$ are two random variables with zero mean and variances $C_{j,4}^{(i,t)}$ and $C_{j,5}^{(i,t)}$, respectively. 
\end{thm}

The intuition to achieve Theorem \ref{thm:noise_QNN_gaussian_formal}   is as follows. As explained in the main text, the discrepancy between the estimated gradient $\nabla_j\bar{\mathcal{L}}_i(\bm{\theta}^{(t)})$ and the analytic gradient $\nabla_j {\mathcal{L}}_i(\bm{\theta}^{(t)})$  is caused by the difference between the estimated results $\bar{Y}_i^{(t)}$ (or $\bar{Y}_i^{(t,\pm_j)}$) and the expected results $\hat{Y}_i^{(t)}$ (or $\hat{Y}_i^{(t,\pm j)}$), due to the involved depolarization noise $\mathcal{N}_{p}$ and the finite number of measurements $K$. Specifically, the noisy channel  $\mathcal{N}_{p}$ shifts the expectation values, and the finite number of measurements $K$ turns  the output of quantum circuit from the determination to be random. Under the above observation, the estimated gradients $ \nabla_j \bar{\mathcal{L}}_i(\bm{\theta}^{(t)})$ can be treated as the random variable that is formed by three random variables $\bar{Y}_i^{(t)}$ and $\bar{Y}_i^{(t,\pm_j)}$, where the probability distributions of $\bar{Y}_i^{(t)}$ and $\bar{Y}_i^{(t,\pm_j)}$ are determined by $K$, $\mathcal{N}_{p}$,   $\hat{Y}_i^{(t)}$, and $\hat{Y}_i^{(t,\pm j)}$. Therefore, to explicitly build the relation between $ \nabla_j \bar{\mathcal{L}}_i(\bm{\theta}^{(t)})$  and $ \nabla_j {\mathcal{L}}_i(\bm{\theta}^{(t)})$, we should first formulate the distribution of the estimated gradients using $\bar{Y}_i^{(t)}$ and $\bar{Y}_i^{(t,\pm_j)}$, and then connect the obtained distribution with the analytic gradients. The following lemma summarizes the distribution of the estimated gradients using $\bar{Y}_i^{(t)}$ and $\bar{Y}_i^{(t,\pm_j)}$, whose proof is given in Subsection \ref{subsec:Appendix_lem_thm_noise_dist}.       
 
 \begin{lem}\label{lem:estimate_mean_variance}
 	The   mean $\nu_i^{(t)}$ and variance $(\sigma_i^{(t)})^2$ of the  estimated result $\bar{Y}_i^{(t)}$ are   
\begin{align}\label{eqn:thm_gauss_0_1}
& \nu^{(t)}=(1-\tilde{p})\hat{Y}_i^{(t)}+\tilde{p}\frac{\Tr(\Pi)}{D}~, \nonumber\\ & (\sigma_i^{(t)})^2 = \frac{-(1-\tilde{p})^2(\hat{Y}_i^{(t)})^2 + (1-\tilde{p})\left(1-2\tilde{p}\frac{\Tr(\Pi)}{D}\right)\hat{Y}_i^{(t)}+\tilde{p}\frac{\Tr(\Pi)}{D}- \tilde{p}^2\frac{(\Tr(\Pi))^2}{D^2}}{K}~.
\end{align}
The  mean $\nu_i^{(t,\pm_j)}$ and variance $(\sigma_i^{(t,\pm_j)})^2$ of the estimated results $\bar{Y}_i^{(t,\pm_j)}$ are 
\begin{align}\label{eqn:thm_gauss_1}
	& \nu^{(t,\pm_j)}=(1-\tilde{p})\hat{Y}_i^{(t,\pm_j)}+\tilde{p}\frac{\Tr(\Pi)}{D}~, \nonumber\\
	& (\sigma_i^{(t,\pm_j)})^2 = \frac{-(1-\tilde{p})^2(\hat{Y}_i^{(t,\pm_j)})^2 + (1-\tilde{p})\left(1-2\tilde{p}\frac{\Tr(\Pi)}{D}\right)\hat{Y}_i^{(t,\pm_j)}+\tilde{p}\frac{\Tr(\Pi)}{D}- \tilde{p}^2\frac{(\Tr(\Pi))^2}{D^2}}{K}~.
\end{align}

 \end{lem}

\begin{proof}[Proof of Theorem  \ref{thm:noise_QNN_gaussian_formal}]
We now utilize the established relations as shown  in Lemma \ref{lem:estimate_mean_variance}   to obtain the relation between the estimated and the analytic gradients. Recall that, at the $t$-th iteration, given the input $\mathcal{B}_i$ and $K$ measurements, the estimated gradient for $j$-th parameter $\bm{\theta}_j$ of noisy QNN  is    
\begin{equation}\label{eqn:thm_gauss_0}
	\nabla_j \bar{\mathcal{L}}_i(\bm{\theta}^{(t)}) =  (\bar{Y}_i^{(t)} - {Y}_i)\left(\bar{Y}_i^{(t,+_j)} - \bar{Y}_i^{(t,-_j)}\right) + \lambda \bm{\theta}_j^{(t)}~.
\end{equation}

Combining Lemma \ref{lem:estimate_mean_variance} and Eqn.~(\ref{eqn:thm_gauss_0}), the term $\Delta_i^{(t,j)}:=\bar{Y}_i^{(t,+_j)} - \bar{Y}_i^{(t,-_j)}$ in Eqn.~(\ref{eqn:thm_gauss_0}) can be treated as the difference of two random variables. The term $(\bar{Y}_i^{(t)} - {Y}_i)$ in Eqn.~(\ref{eqn:thm_gauss_0})  can also be treated as a random variables. We now separately investigate their moment properties. 

\textit{\underline{The term $\Delta_i^{(t,j)}$.}} Following the notations used in Lemma \ref{lem:estimate_mean_variance},  the mean and variance of the term $\Delta_i^{(t,j)}$ are  $\nu_i^{(t,+_j)}-\nu_i^{(t,-_j)}$ and $(\sigma_i^{(t,j)})^2 =(\sigma_i^{(t,+_j)})^2 + (\sigma_i^{(t,-_j)})^2$, supported by the definition of  moments and the independent relation between $\bar{Y}_i^{(t,+_j)}$ and $\bar{Y}_i^{(t,-_j)}$.
   
By leveraging the explicit form of $\nu_i^{(t,\pm_j)}$, the random variable $\Delta_i^{(t,j)}$ can be rewritten as
\begin{equation}\label{eqn:thm_gauss_2_1}
	\Delta_i^{(t,j)} = (1-\tilde{p})(\hat{Y}^{(t,+_j)}-\hat{Y}^{(t,-_j)})+\xi^{(t,j)}~,
\end{equation}
where $\xi^{(t,j)}$ is a random variable with zero mean and variance $(\sigma_i^{(t,j)})^2$. 

\textit{\underline{The term $(\bar{Y}_i^{(t)} - {Y}_i)$.}} Following the notations used in Lemma \ref{lem:estimate_mean_variance},  an equivalent representation of $(\bar{Y}_i^{(t)}-\bar{Y}_i^{(t)})$ is 
\begin{equation}\label{eqn:thm_gauss_2_2}
(\bar{Y}_i^{(t)}-\bar{Y}_i^{(t)})=	(1-\tilde{p})\hat{Y}_i^{(t)} +\tilde{p}\frac{\Tr(\Pi)}{D} + \xi^{(t)} -\bar{Y}_i^{(t)}~,
\end{equation}
where $\xi^{(t)}$ is a random variable with zero mean and variance $(\sigma_i^{(t)})^2$. 

The reformulated terms as shown in Eqn.~(\ref{eqn:thm_gauss_2_1}) and Eqn.~(\ref{eqn:thm_gauss_2_2}) indicate that the estimated result $\nabla_j \bar{\mathcal{L}}_i(\bm{\theta}^{(t)})$  can be rewritten as
\begingroup
\allowdisplaybreaks
\begin{align}
& \nabla_j \bar{\mathcal{L}}_i(\bm{\theta}^{(t)})	\nonumber\\
= &(\bar{Y}_i^{(t)}-{Y}_i)(\bar{Y}_i^{(t,+_j)} - \bar{Y}_i^{(t,-_j)}) + \lambda\bm{\theta}_j^{(t)} \nonumber\\
	= & \left((1-\tilde{p})\hat{Y}_i^{(t)} +\tilde{p}\frac{\Tr(\Pi)}{D} -{Y}_i\right)(1-\tilde{p})(\hat{Y}^{(t,+_j)}-\hat{Y}^{(t,-_j)}) + \left((1-\tilde{p})\hat{Y}_i^{(t)} +\tilde{p}\frac{\Tr(\Pi)}{D}-{Y}_i\right)\xi^{(t,j)}\nonumber\\
	&+ 	(1-\tilde{p})(\hat{Y}^{(t,+_j)}-\hat{Y}^{(t,-_j)})\xi^{(t)}+ \xi^{(t)}\xi^{(t,j)}  + \lambda\bm{\theta}_j^{(t)} \nonumber\\
	= & (1-\tilde{p})^2 \nabla_j\mathcal{L}_i(\bm{\theta}^{(t)}) + (1-\tilde{p})\tilde{p}\left( \frac{\Tr(\Pi)}{D}-{Y}_i\right)(\hat{Y}^{(t,+_j)}-\hat{Y}^{(t,-_j)}) + (2\tilde{p}-\tilde{p}^2)\lambda\bm{\theta}_j^{(t)} \nonumber\\
	&  + (1-\tilde{p})(\hat{Y}^{(t,+_j)}-\hat{Y}^{(t,-_j)})\xi^{(t)} + \left((1-\tilde{p})\hat{Y}_i^{(t)} +\tilde{p}\frac{\Tr(\Pi)}{D}-{Y}_i\right)\xi^{(t,j)} + \xi^{(t)}\xi^{(t,j)}~.
\end{align}
\endgroup
Combining the above equation and the explicit expression of $\xi^{(t)}$ and $\xi^{(t,j)}$, we obtain the relation between the estimated and the analytic gradients. Specifically, the estimated gradient can be formulated as \[\nabla_j\bar{\mathcal{L}}_i(\bm{\theta}^{(t)}) = (1-\tilde{p})^2\nabla_j{\mathcal{L}}_i(\bm{\theta}^{(t)}) + C_{j,1}^{(i,t)} +  \bm{\varsigma}_i^{(t,j)}~,\]
where $\bm{\varsigma}_i^{(t,j)}=  C_{j,2}^{(i,t)}\xi_i^{(t)}+ C_{j,3}^{(i,t)}\xi_i^{(t,j)} + \xi^{(t)}\xi_i^{(t,j)}$, the first three constants $\{C_{j,1}^{(i,t)}\}_{i=1}^3$ are defined as 
\[
    C^{(i,t)}_{j,a}= 
\begin{cases}
   (1-\tilde{p})\tilde{p}\left(\frac{\Tr(\Pi)}{D}-{Y}_i\right)(\hat{Y}^{(t,+_j)}-\hat{Y}^{(t,-_j)}) + (2\tilde{p}-\tilde{p}^2)\lambda\bm{\theta}_j^{(t)}~,&  a=1\\
    (1-\tilde{p})(\hat{Y}_i^{(t,+_j)}-\hat{Y}_i^{(t,-_j)})~,  & a=2 \\
\left((1-\tilde{p})\hat{Y}_i^{(t)} +\tilde{p}\frac{\Tr(\Pi)}{D}-{Y}_i\right)~,   & a=3 ~,
\end{cases}
\]
and the last two constants, which  separately correspond to the variance $(\sigma_i^{(t)})^2$ and $(\sigma_i^{(t,j)})^2$ of the random variables $\xi_i^{(t)}$ and $\xi_i^{(t,j)}$, are
\[
    C^{(i,t)}_{j,a}= 
\begin{cases}
  \frac{-(1-\tilde{p})^2(\hat{Y}_i^{(t)})^2 + (1-\tilde{p})\left(1-2\tilde{p}\frac{\Tr(\Pi)}{D}\right)\hat{Y}_i^{(t)}+\tilde{p}\frac{\Tr(\Pi)}{D}- \tilde{p}^2\frac{(\Tr(\Pi))^2}{D^2}}{K}~,&  a=4\\
\frac{-(1-\tilde{p})^2((\hat{Y}_i^{(t,+_j)})^2+(\hat{Y}_i^{(t,-_j)})^2) + (1-\tilde{p})\left(1-2\tilde{p}\frac{\Tr(\Pi)}{D}\right)(\hat{Y}_i^{(t,+_j)}+\hat{Y}_i^{(t,-_j)})+2\tilde{p}\frac{\Tr(\Pi)}{D}- 2\tilde{p}^2\frac{(\Tr(\Pi))^2}{D^2}}{K}~,  & a=5~.
\end{cases}
\]

\end{proof}

\subsection{Proof of Lemma \ref{lem:estimate_mean_variance}}\label{subsec:Appendix_lem_thm_noise_dist}

To achieve Lemma \ref{lem:estimate_mean_variance},  we first simplify the learning model of QNN with the  depolarization noise.  In particular,  all noisy channels $\mathcal{N}_{p}$, which are separately applied to each quantum circuit depth, can be merged together to a specific circuit depth and presented by a new depolarization channel $\mathcal{N}_{\tilde{p}}$.  
\begin{lem}\label{lem:equi_dep}
	Let  $\mathcal{N}_p$  be the depolarization channel. There always exists a depolarization channel $\mathcal{N}_{\tilde{p}}$ with $\tilde{p}=1 - (1-p)^{L_Q}$ that satisfies $  \mathcal{N}_p({U}_L(\bm{\theta})...{U}_2(\bm{\theta})\mathcal{N}_p({U}_1(\bm{\theta})\rho {U}_1(\bm{\theta})^{\dagger}){U}_2(\bm{\theta})^{\dagger}...{U}_L(\bm{\theta})^{\dagger}) =\mathcal{N}_{\tilde{p}} ({U}(\bm{\theta})  \rho {U}(\bm{\theta})^{\dagger})$, where $\rho$ is the input quantum state.
\end{lem}

\begin{proof}[Proof of Lemma \ref{lem:equi_dep}]
Denote $\rho^{(k)}$ as $\rho^{(k)} = \prod_{l=1}^k U_{l}(\bm{\theta})\rho U_{l}(\bm{\theta})^{\dagger}$. Applying $\mathcal{N}_p$ to $\rho^{(1)}$ gives
\begin{equation}
	\mathcal{N}_p(\rho^{(1)}) = (1-p)\rho^{(1)} + p\frac{\mathbb{I}_D}{D}~,
\end{equation}
where $D$ refers to the dimensions of Hilbert space interacted with $\mathcal{N}_p$. 
 
Supporting by the above equation, applying $U_2(\bm{\theta})$ to the state $\mathcal{N}_p(\rho^{(1)})$ gives
\begin{equation}
	U_2{(\bm{\theta})} \mathcal{N}_p(\rho^{(1)}) U_2{(\bm{\theta})}^{\dagger} = (1-p)\rho^{(2)} + p\frac{\mathbb{I}_D}{D}~. 
\end{equation}
Then interacting $\mathcal{N}_p$ with the state $U_2{(\bm{\theta})} \mathcal{N}_p(\rho^{(1)}) U_2{(\bm{\theta})}^{\dagger}$    gives
\begin{equation}
	\mathcal{N}_p (U_2{(\bm{\theta})} \mathcal{N}_p (\rho^{(1)}) U_2{(\bm{\theta})}^{\dagger} ) = (1-p)^2 \rho^{(2)} + (1-p)p\frac{\mathbb{I}_D}{D} + p\frac{\mathbb{I}_D}{D} = (1-p)^2 \rho^{(2)} + (1-(1-p)^2)\frac{\mathbb{I}_D}{D} ~.
\end{equation}
 
By induction, suppose at $k$-th step, the generated state is 
	\begin{equation}
		\rho^{(k)} = (1-p)^l\rho^{(k)} + (1-(1-p)^k)\frac{\mathbb{I}_D}{D}~.
	\end{equation}
Then applying $U_{k+1}(\bm{\theta})$ followed by $\mathcal{N}_p$ gives
\begin{equation}
	\rho^{(k+1)} =\mathcal{N}_p\left(U_{k+1}(\bm{\theta})\rho^{(k)}U_{k+1}(\bm{\theta})^{\dagger} \right) = (1-p)^{k+1}\rho^{(k+1)} + (1-(1-p)^{k+1})\frac{\mathbb{I}_D}{D}~.
\end{equation} 
According to the formula of depolarization channel, an immediate observation is that the noisy QNN is equivalent to applying a single depolarization channel $\mathcal{N}_{\tilde{p}}$ at the last circuit depth $L$, i.e.,
\begin{equation}
	\mathcal{N}_{\tilde{p}}(\rho) = (1-p)^{L} \rho^{(L)} + (1-(1-p)^L)\frac{\mathbb{I}}{D}~,
\end{equation}
where 
\begin{equation}
	\tilde{p} =1- (1-p)^L~.
\end{equation}
\end{proof}

We then use the simplified QNN given by Lemma \ref{lem:equi_dep} to explore the relation between the generated statistic $\bar{Y}_i^{(t)}$ and the expectation value $\hat{Y}^{(t)}$ (the same rule applies to connect $\bar{Y}_i^{(t,\pm_j)}$ with $\hat{Y}^{(t,\pm_j)}$).
 
At the $t$-th iteration, given the tunable parameters $\bm{\theta}^{(t)}$ and inputs $\mathcal{B}_i$, the ensemble corresponding to the  generated state of QNN before taking quantum measurements is  $\{p_l, \gamma_{i,l}^{(t)}\}_{l=1}^2$, i.e., $p_1=1 - \tilde{p}$ with $\gamma_{i,1}^{(t)}={U}(\bm{\theta}^{(t)})\rho_{\mathcal{B}_i} {U}(\bm{\theta}^{(t)})^{\dagger}$ and $p_2=\tilde{p}$ with $\gamma_{i,2}^{(t)}=\mathbb{I}_D/D$.  After applying a two-outcome POVM $\Pi$ to measure  such an ensemble $K$ times, the generated statistics (sample mean)  is $\bar{Y}_i^{(t)}=\frac{1}{K}\sum_{k=1}^K V_k^{(t)}$, where each measured outcome $V_k^{(t)}$ with $k\in[K]$ is a random variable that satisfies Fact \ref{fact:dist_QNN}.  
  \begin{fact}\label{fact:dist_QNN}
$V_k^{(t)}$ is a random variable that follows the  distribution $\mathcal{P}_{Q'}(V_k^{(t)}) = \sum_{c=1}^2 \Pr(z=c)\Pr(V_k^{(t)}|z=c)$. The explicit formula of $\mathcal{P}_{Q'}$ is 
\begin{enumerate}
	\item $\Pr(z=1)=1-\tilde{p}$ with   $V_k^{(t)}|z=1 \sim \Ber(\hat{Y}^{(t)}_i)$ and $\hat{Y}^{(t)}_i=\Tr(\Pi \gamma_{i,1}^{(t)})$~;
	\item $\Pr(z=2)=\tilde{p}$ with   $V_k^{(t)}|z=2 \sim \Ber(\frac{\Tr(\Pi)}{D})$ with $\frac{\Tr(\Pi)}{D}=\Tr(\Pi \gamma_{i,2}^{(t)})$~.
\end{enumerate} 
 \end{fact}
Fact \ref{fact:dist_QNN} implies that the mean and variance of $V_k^{(t)}$ are \[(1-\tilde{p})\hat{Y}_i^{(t)}+\tilde{p}\frac{\Tr(\Pi)}{D} \text{and}~ -(1-\tilde{p})^2(\hat{Y}_i^{(t)})^2 + (1-\tilde{p})\left(1-2\tilde{p}\frac{\Tr(\Pi)}{D}\right)\hat{Y}_i^{(t)}+\tilde{p}\frac{\Tr(\Pi)}{D}- \tilde{p}^2\frac{(\Tr(\Pi))^2}{D^2}~,\] respectively. Moreover, since each outcome $V_k^{(t)}$ follows the distribution $\mathcal{P}_{Q'}$, the mean $\nu_i^{(t)}$  and the variance $(\sigma_i^{(t)})^2$ of the  sample mean $\bar{Y}_i^{(t)}$ are   
\begin{align}
& \nu^{(t)}=(1-\tilde{p})\hat{Y}_i^{(t)}+\tilde{p}\frac{\Tr(\Pi)}{D}~, \nonumber\\ & (\sigma_i^{(t)})^2 = \frac{-(1-\tilde{p})^2(\hat{Y}_i^{(t)})^2 + (1-\tilde{p})\left(1-2\tilde{p}\frac{\Tr(\Pi)}{D}\right)\hat{Y}_i^{(t)}+\tilde{p}\frac{\Tr(\Pi)}{D}- \tilde{p}^2\frac{(\Tr(\Pi))^2}{D^2}}{K}~.
\end{align}

Following the same routine,  where  the mean $\nu_i^{(t,\pm_j)}$ and the variance $(\sigma_i^{(t,\pm_j)})^2$ of the sample mean $\bar{Y}_i^{(t,\pm_j)}$ satisfy 
\begin{align}
	& \nu^{(t,\pm_j)}=(1-\tilde{p})\hat{Y}_i^{(t,\pm_j)}+\tilde{p}\frac{\Tr(\Pi)}{D}~, \nonumber\\
	& (\sigma_i^{(t,\pm_j)})^2 = \frac{-(1-\tilde{p})^2(\hat{Y}_i^{(t,\pm_j)})^2 + (1-\tilde{p})\left(1-2\tilde{p}\frac{\Tr(\Pi)}{D}\right)\hat{Y}_i^{(t,\pm_j)}+\tilde{p}\frac{\Tr(\Pi)}{D}- \tilde{p}^2\frac{(\Tr(\Pi))^2}{D^2}}{K}~.
\end{align}


\section{Proof of Theorem  \ref{thm:informal_utl_QNNQAE_DP}}\label{Appendix:Thm_utl_QNN}
Theorem  \ref{thm:informal_utl_QNNQAE_DP} quantifies  the utility bounds  $R_1$ and $R_2$  of QNN under the depolarization noise towards ERM framework.  For ease of illustration, we restate Theorem \ref{thm:informal_utl_QNNQAE_DP} below. 
\begin{thm}[Restate of Theorem \ref{thm:informal_utl_QNNQAE_DP}] \label{thm:informal_utl_QNNQAE_DP_restate}
QNN outputs $\bm{\theta}^{(T)}\in\mathbb{R}^d$ after $T$ iterations with utility bounds  $R_1 \leq  \tilde{O}\left(d, \frac{1}{BK}, \frac{1}{(1-p)^{L_Q}} \right)$ and $R_2\leq \tilde{O}\left(\frac{1}{K^2B}, d, \frac{1}{(1-p)^{L_Q}} \right)$, where $K$ is the number of quantum measurements, $L_Q$ is the quantum circuit depth, $p$ is the gate noise, and $B$ is the number of batches. 
\end{thm}

The high level idea to achieve the utility bounds $R_1$ and $R_2$ is as follows. Recall that  $R_1$   measures how far the trainable parameter of QNN is away from the stationary point. A well-known result in optimization theory  \cite{jin2017escape} is that when a function satisfies the smooth property, its  stationary point can be efficiently located by a simple gradient-based algorithm. By leveraging this observation and the relation between the estimated and analytic gradients as achieved in  Theorem \ref{thm:noise_QNN_gaussian_formal}, we can quantify how the estimated gradients of QNN converge to the stationary point, which corresponds to the utility bound $R_1$.  

Recall that the utility bound $R_2$ evaluates the disparity between the expected empirical risk and the optimal risk that is determined by the global minimum. To achieve $R_2$, we utilize the result of  the study \cite{nesterov2006cubic}, which claims that if a non-convex  function satisfies PL condition, then every stationary point is the global minimum. Since the objective function used in QNN satisfies PL condition as shown in Lemma \ref{lem:Lsmmoth}, we can effectively combine the PL condition with the result of $R_1$ to obtain the utility bound $R_2$.

\begin{proof}[Proof of  Theorem \ref{thm:informal_utl_QNNQAE_DP_restate}]
	We employ the following two theorems to achieve Theorem \ref{thm:informal_utl_QNNQAE_DP_restate}, whose proofs are given in   Subsections \ref{Append:subsec:R1_QNN} and  \ref{Append:subsec:R2_QNN},  respectively.
\begin{thm}\label{thm:erm_QNN_r1}
Given the dataset $\bm{z}$,  QNN outputs $\bm{\theta}^{(T)}$ after $T$ iterations with utility bound \[
 	R_1 \leq \frac{ 2S(1 + 90 \lambda d)}{T(1-\tilde{p})^2}   +\frac{(2\tilde{p}-\tilde{p}^2)(2G + d)(1+10\lambda)^2}{(1-\tilde{p})^2}  +  \frac{6dK+8d}{(1-\tilde{p})^2BK^2} ~.\]
\end{thm}

\begin{thm}\label{thm:exp_erm_QNN}
Given the dataset $\bm{z}$,  QNN outputs $\bm{\theta}^{(T)}$ after $T$ iterations with utility bound \[
 	R_2 \leq (1+90\lambda d) \exp\left(-\frac{\mu(1-\tilde{p})^2T}{S} \right)  +  T\frac{(2\tilde{p}-\tilde{p}^2)(G + 2d)(1+10\lambda)^2BK^2+ 6dK+8d}{2SBK^2} .\]
\end{thm}

As for $R_1$, with setting $T\leftarrow \infty$ and after the simplification, the utility bound as shown in Theorem \ref{thm:erm_QNN_r1} follows 
  \begin{align}
  	R_1 \leq \tilde{O}\left( \frac{1}{BK}, \frac{1}{(1-p)^{L_Q}}, d \right)~.
  \end{align}
  
  As for $R_2$, with  setting $T=\mathcal{O}\left(\frac{S}{\mu(1-\tilde{p})^2}\ln\left( \frac{(1+90\lambda d)2SBK^2}{(2\tilde{p}-\tilde{p}^2)(G + 2d)(1+10\lambda)^2BK^2+ 6dK+8d} \right)\right)$ and after simplification, the utility bound as shown in Theorem \ref{thm:exp_erm_QNN} follows 
 \begin{align}
 	R_2 & \leq  \tilde{O}\left(\frac{1}{K^2B}, d, \frac{1}{(1-p)^{L_Q}} \right)  ~.
 \end{align}
	 	
\end{proof}

\subsection{Proof of Theorem \ref{thm:erm_QNN_r1}: The utility bound $R_1$}\label{Append:subsec:R1_QNN}
The proof of  Theorem \ref{thm:erm_QNN_r1} employs the following Lemma, where its proof   is given in Subsection \ref{Append:subsec:R1_QNN_lem}.
\begin{lem}\label{lem:proof_thm_utility_R1_QNN}
	Taking expectation over the randomness of $\xi_i^{(t)}$ and $\xi_i^{(t,j)}$ in the estimated gradient $\nabla_j \bar{\mathcal{L}}(\bm{\theta}^{(t)}$ as formulated in Theorem \ref{thm:noise_QNN_gaussian_formal}, the term $\frac{1}{2S}\sum_{j=1}^d  \mathbb{E}_{\xi_i^{(t)},\xi_i^{(t,j)}} \left[ \left(\nabla_j \bar{\mathcal{L}}(\bm{\theta}^{(t)})   \right)^2 \right] $ with $S$ being the smooth parameter is upper bounded by 
	\[  \frac{(1-\tilde{p})^4 }{2S} \|  \nabla {\mathcal{L}}(\bm{\theta}^{(t)}) \|^2 + \frac{(1-\tilde{p})^2 G}{2S }\max_{i,j} C_{j,1}^{(i,t)} + \frac{d}{2S }\max_{i,j} \left(C_{j,1}^{(i,t)}\right)^2 + \frac{6dK+8d}{2SBK^2}~. \]
\end{lem}

\begin{proof}[Proof of Theorem \ref{thm:erm_QNN_r1}]
			
Recall that the optimization rule of noisy QNN at the $t$-th iteration follows
\begin{align}\label{eqn:thm_erm_QNN_0}
	\bm{\theta}^{(t+1)} = 	\bm{\theta}^{(t)} -   \eta  \nabla\bar{\mathcal{L}}(\bm{\theta}^{(t)}) ~. 
\end{align}

Since the objective function $\mathcal{L}(\bm{\theta})$ is $S$-smooth, as indicated in Lemma \ref{lem:Lsmmoth}, we have
\begin{equation}\label{eqn:thm_emr_QNN_1}
	\mathcal{L}(\bm{\theta}^{(t+1)}) -\mathcal{L}(\bm{\theta}^{(t)}) \leq \langle\nabla \mathcal{L}(\bm{\theta}^{(t)}), \bm{\theta}^{(t+1)}-\bm{\theta}^{(t)} \rangle +\frac{S}{2}\|\bm{\theta}^{(t+1)}-\bm{\theta}^{(t)}\|^2~.
\end{equation}
Combine the above two equations and setting $\eta = 1/S$, we have
\begin{align}\label{eqn:thm_emr_QNN_2}
	& \mathcal{L}(\bm{\theta}^{(t+1)}) -\mathcal{L}(\bm{\theta}^{(t)}) \nonumber\\
	 \leq & \langle\nabla \mathcal{L}(\bm{\theta}^{(t)}), \bm{\theta}^{(t+1)}-\bm{\theta}^{(t)} \rangle +\frac{S}{2}\|\bm{\theta}^{(t+1)}-\bm{\theta}^{(t)}\|^2 \nonumber\\
	= & -\frac{1}{S}\langle\nabla \mathcal{L}(\bm{\theta}^{(t+1)}),  \nabla \bar{\mathcal{L}}(\bm{\theta}^{(t)})   \rangle +\frac{1}{2S}\|\nabla \bar{\mathcal{L}}(\bm{\theta}^{(t)})  \|^2 \nonumber\\
	 = &  -\frac{1}{S}\sum_{j=1}^d \left( \nabla_j \mathcal{L}(\bm{\theta}^{(t+1)})\nabla_j  \bar{\mathcal{L}}(\bm{\theta}^{(t)})   \right) +\frac{1}{2S}\sum_{j=1}^d  \left(\nabla_j \bar{\mathcal{L}}(\bm{\theta}^{(t)})   \right)^2 ~.
\end{align}

Recall the definition of the estimated gradient is $\nabla_j \bar{\mathcal{L}}(\bm{\theta}^{(t)})=\frac{1}{B}\sum_{i=1}^B \nabla_j \bar{\mathcal{L}}_i(\bm{\theta}^{(t)})$ and the explicit expression of $\nabla_j \bar{\mathcal{L}}_i(\bm{\theta}^{(t)})$ is  \[\nabla_j \bar{\mathcal{L}}_i(\bm{\theta}^{(t)})= (1-\tilde{p})^2\nabla_j{\mathcal{L}}_i(\bm{\theta}^{(t)}) + C_{j,1}^{(i,t)} + C_{j,2}^{(i,t)}\xi^{(t)} + C_{j,3}^{(i,t)}\xi_{i}^{(t,j)}+\xi_{i}^{(t)}\xi_{i}^{(t,j)}~.\] 
Alternatively, the gradient for the $j$-th parameter $\nabla_j \bar{\mathcal{L}}(\bm{\theta}^{(t)})$ follows     
\begin{equation}\label{eqn:thm_emr_QNN_2_1}
	\nabla_j \bar{\mathcal{L}}(\bm{\theta}^{(t)}) =\frac{1}{B}\sum_{i=1}^B (1-\tilde{p})^2\nabla_j{\mathcal{L}}_i(\bm{\theta}^{(t)}) + C_{j,1}^{(i,t)} + C_{j,2}^{(i,t)}\xi_{i}^{(t)} + C_{j,3}^{(i,t)}\xi_{i}^{(t,j)}+\xi_{i}^{(t)}\xi^{(t,j)}~.
\end{equation}

Combining Eqn.~(\ref{eqn:thm_emr_QNN_2}) with Eqn.~(\ref{eqn:thm_emr_QNN_2_1})   and taking  expectation over $\xi_i^{(t)}$ and $\xi_i^{(t,j)}$, we obtain  
\begingroup
\allowdisplaybreaks
\begin{align}\label{eqn:thm_emr_QNN_3}
	& \mathbb{E}_{\xi_i^{(t)},\xi_i^{(t,j)}}[\mathcal{L}(\bm{\theta}^{(t+1)}) -\mathcal{L}(\bm{\theta}^{(t)})] \nonumber\\
\leq & -\frac{1}{S}(1-\tilde{p})^2 \|\nabla \mathcal{L}(\bm{\theta}^{(t)}) \|^2 -\frac{1}{S}\sum_{j=1}^d \nabla_j\mathcal{L}(\bm{\theta}^{(t)})\left(\frac{1}{B}\sum_{i=1}^B C_{j,1}^{(i,t)} \right)\nonumber\\
 & -\frac{1}{S}\sum_{j=1}^d \nabla_j\mathcal{L}(\bm{\theta}^{(t)}) \frac{1}{B}\sum_{i=1}^B\mathbb{E}_{\xi_i^{(t)}}\left[C_{j,2}^{(i,t)}\xi_i^{(t)} \right] -\frac{1}{S}\sum_{j=1}^d \nabla_j\mathcal{L}(\bm{\theta}^{(t)}) \frac{1}{B}\sum_{i=1}^B\mathbb{E}_{\xi_i^{(t,j)}}\left[C_{j,3}^{(i,t)}\xi_i^{(t,j)} \right] \nonumber\\
 & -  \frac{1}{S}\sum_{j=1}^d \nabla_j\mathcal{L}(\bm{\theta}^{(t)}) \frac{1}{B}\sum_{i=1}^B\mathbb{E}_{\xi_i^{(t)},\xi_i^{(t,j)}}\left[\xi_i^{(t)} \xi_i^{(t,j)} \right] + \frac{1}{2S}\sum_{j=1}^d  \mathbb{E}_{\xi_i^{(t)},\xi_i^{(t,j)}} \left[ \left(\nabla_j \bar{\mathcal{L}}(\bm{\theta}^{(t)})   \right)^2 \right] \nonumber\\
 \leq & -\frac{1}{S}(1-\tilde{p})^2 \|\nabla \mathcal{L}(\bm{\theta}^{(t)}) \|^2 + \frac{G}{2S} \max_{i,j} C_{j,1}^{(i,t)}  +   \frac{1}{2S}\sum_{j=1}^d  \mathbb{E}_{\xi_i^{(t)},\xi_i^{(t,j)}} \left[ \left(\nabla_j \bar{\mathcal{L}}(\bm{\theta}^{(t)})   \right)^2 \right]~.
 \end{align}
 \endgroup
The first inequality uses the result of  Eqn.~(\ref{eqn:thm_emr_QNN_2_1}). The second inequality uses  $\mathbb{E}[\xi^{(t)}_i]=0$, $\mathbb{E}[\xi^{(t,j)}_i]=0$ as shown in Theorem \ref{thm:noise_QNN_gaussian_formal}, and $- G/d \leq \nabla_j \mathcal{L}(\bm{\theta}^{(t)}) \leq G/d$ supported by $G$-Lipschitz property.

By leveraging Lemma \ref{lem:proof_thm_utility_R1_QNN},  Eqn.~(\ref{eqn:thm_emr_QNN_3})   can be further simplified as 
\begin{align}\label{eqn:thm_emr_QNN_3_2}
	& \mathbb{E}_{\xi_i^{(t)},\xi_i^{(t,j)}}[\mathcal{L}(\bm{\theta}^{(t+1)}) -\mathcal{L}(\bm{\theta}^{(t)})] \nonumber\\
	\leq &  -\frac{1}{S}(1-\tilde{p})^2 \|\nabla \mathcal{L}(\bm{\theta}^{(t)}) \|^2 +\frac{G}{2S}\max_{i,j} C_{j,1}^{(i,t)} + \frac{(1-\tilde{p})^4 }{2SB} \|  \nabla_j{\mathcal{L}}(\bm{\theta}^{(t)}) \|^2  \nonumber\\
	& + \frac{(1-\tilde{p})^2 G}{2S}\max_{i,j} C_{j,1}^{(i,t)} + \frac{d}{2S}\max_{i,j} \left( C_{j,1}^{(i,t)}\right)^2 + \frac{6dK+8d}{2SBK^2} \nonumber\\
	\leq & -\frac{1}{2S}(1-\tilde{p})^2 \|\nabla \mathcal{L}(\bm{\theta}^{(t)}) \|^2   +\frac{2G+d}{2S}(2-\tilde{p})\tilde{p}(1+10\lambda)^2  + \frac{6dK+8d}{2SBK^2}~.   
\end{align}
The first inequalities comes from  Lemma \ref{lem:proof_thm_utility_R1_QNN}, and the second inequality employs $ \frac{(1-\tilde{p})^4 }{2SB}    \leq  \frac{(1-\tilde{p})^2 }{2S} $ and the following result
 \begin{align}
 	& \frac{G}{2S}  \max_{i,j} C_{j,1}^{(i,t)}  + \frac{(1-\tilde{p})^2 G}{2S}\max_{i,j} C_{j,1}^{(i,t)} + \frac{d}{2S}\max_{i,j} \left( C_{j,1}^{(i,t)}\right)^2 \nonumber\\
 	\leq & \frac{(1+(1-\tilde{p})^2)G}{2S}(2-\tilde{p})\tilde{p}(1+10\lambda) + \frac{d}{2S} (2-\tilde{p})\tilde{p}(1+10\lambda)^2  \nonumber\\
 	\leq & \frac{2G+d}{2S}(2-\tilde{p})\tilde{p}(1+10\lambda)^2 ~,
 \end{align}
where the first inequality uses the upper bound of $C_{j,1}^{(i,t)}$ and $(C_{j,1}^{(i,t)})^2$, i.e.,
$\max_{i,j} C_{j,1}^{(i,t)}  \leq (1-\tilde{p})\tilde{p} + 10(2-\tilde{p})\tilde{p}\lambda \leq (2-\tilde{p})\tilde{p}(1+10\lambda)$ and $\max_{i,j} \left(C_{j,1}^{(i,t)}\right)^2  \leq \left((2-\tilde{p})\tilde{p}(1+10\lambda)\right)^2\leq (2-\tilde{p})\tilde{p}(1+10\lambda)^2$, and the second inequality uses $(1-\tilde{p})^2\leq 1$.

 An equivalent representation of Eqn.~(\ref{eqn:thm_emr_QNN_3_2}) is 
 \begin{align}\label{eqn:l2_risk_QNN_3_2}
\|\nabla \mathcal{L}(\bm{\theta}^{(t)}) \|^2 \leq  2S\frac{ \mathcal{L}(\bm{\theta}^{(t)}) - \mathbb{E}_{\xi_i^{(t)},\xi_i^{(t,j)}}[\mathcal{L}(\bm{\theta}^{(t+1)})]}{(1-\tilde{p})^2}   +\frac{(2\tilde{p}-\tilde{p}^2)(2G + d)(1+10\lambda)^2}{(1-\tilde{p})^2}  +  \frac{6dK+8d}{(1-\tilde{p})^2BK^2}~.
 \end{align}
 By induction, with summing over $t=0,...,T-1$ and taking expectation of Eqn.~(\ref{eqn:l2_risk_QNN_3_2}), we obtain  
 \begin{align}
 &	\mathbb{E}_{t}\left[\|\nabla \mathcal{L}(\bm{\theta}^{(t)}) \|^2 \right] \nonumber\\	
 \leq 	&  2S\frac{ \mathcal{L}(\bm{\theta}^{(0)}) - \mathbb{E}_{\xi_i^{(T)},\xi_i^{(T,j)}}[\mathcal{L}(\bm{\theta}^{(T)})]}{T(1-\tilde{p})^2}   +\frac{(2\tilde{p}-\tilde{p}^2)(2G + d)(1+10\lambda)^2}{(1-\tilde{p})^2}  +  \frac{6dK+8d}{(1-\tilde{p})^2BK^2}\nonumber\\
 	\leq &   \frac{ 2S + 2S\lambda d(3\pi)^2 }{T(1-\tilde{p})^2}   +\frac{(2\tilde{p}-\tilde{p}^2)(2G + d)(1+10\lambda)^2}{(1-\tilde{p})^2}  +  \frac{6dK+8d}{(1-\tilde{p})^2BK^2} \nonumber\\  
 \leq	&   \frac{ 2S(1 + 90 \lambda d)}{T(1-\tilde{p})^2}   +\frac{(2\tilde{p}-\tilde{p}^2)(2G + d)(1+10\lambda)^2}{(1-\tilde{p})^2}  +  \frac{6dK+8d}{(1-\tilde{p})^2BK^2}  ~,
 \end{align}
  where the second inequality uses $\mathcal{L}(\bm{\theta}^{(0)}) - \mathbb{E}_{\xi_i^{(T)},\xi_i^{(T,j)}}[\mathcal{L}(\bm{\theta}^{(T)})]\leq \mathcal{L}(\bm{\theta}^{(0)}) - \mathcal{L}^*$, $\mathcal{L}^*>0$ and $\mathcal{L}(\bm{\theta}^{(0)})\leq 1 + \lambda d (3\pi)^2$.
\end{proof}

\subsection{Proof of  Theorem \ref{thm:exp_erm_QNN}: The utility bound $R_2$}\label{Append:subsec:R2_QNN}
\begin{proof}[Proof of  Theorem \ref{thm:exp_erm_QNN}]
The proof of Theorem \ref{thm:exp_erm_QNN} is similar with that of Theorem \ref{thm:erm_QNN_r1}. In particular, following the same routine, we obtain the result of Eqn.(\ref{eqn:thm_emr_QNN_3_2}), i.e., 
\begin{align}
	& \mathbb{E}_{\xi_i^{(t)},\xi_i^{(t,j)}}[\mathcal{L}(\bm{\theta}^{(t+1)}) -\mathcal{L}(\bm{\theta}^{(t)})] \nonumber\\
	\leq &  -\frac{1}{2S}(1-\tilde{p})^2 \|\nabla \mathcal{L}(\bm{\theta}^{(t)}) \|^2   +\frac{2G+d}{2S}(2-\tilde{p})\tilde{p}(1+10\lambda)^2  + \frac{6dK+8d}{2SBK^2}~.
\end{align}

Then, we call the conclusion of PL condition as formulated in Lemma \ref{lem:Lsmmoth} and acquire
\begin{align}\label{thm:emr_risk_0}
	& \mathbb{E}_{\xi_i^{(t)},\xi_i^{(t,j)}}[\mathcal{L}(\bm{\theta}^{(t+1)}) -\mathcal{L}(\bm{\theta}^{(t)})] \nonumber\\
	\leq & -\frac{\mu(1-\tilde{p})^2}{S}(\mathcal{L}(\bm{\theta}^{(t)}) - \mathcal{L}^*)   +\frac{2G+d}{2S}(2-\tilde{p})\tilde{p}(1+10\lambda)^2  + \frac{6dK+8d}{2SBK^2}~. 
\end{align}

An equivalent reformulation of Eqn.~(\ref{thm:emr_risk_0}) is
\begin{align}
	& \mathbb{E}_{\bm{\varsigma}^{(t)}}[\mathcal{L}(\bm{\theta}^{(t+1)})] - \mathcal{L}^* \nonumber \\
	\leq & \left(1-\frac{\mu(1-\tilde{p})^2}{S} \right)(\mathcal{L}(\bm{\theta}^{(t)})-\mathcal{L}^* ) +\frac{2G+d}{2S}(2-\tilde{p})\tilde{p}(1+10\lambda)^2  + \frac{6dK+8d}{2SBK^2}~.
\end{align}	
By induction, with summing over $t=0,...,T$ and taking expectation, we obtain 
\begin{align}
	& \mathbb{E}_{\bm{\varsigma}^{(t)}}[\mathcal{L}(\bm{\theta}^{(T)})] - \mathcal{L}^* \nonumber\\
	\leq  & \left(1-\frac{\mu(1-\tilde{p})^2}{S} \right)^T(\mathcal{L}(\bm{\theta}^{(0)})-\mathcal{L}^* )  + T\frac{2G+d}{2S}(2-\tilde{p})\tilde{p}(1+10\lambda)^2  + T\frac{6dK+8d}{2SBK^2} \nonumber \\
	 \leq & (1+90\lambda d) \exp\left(-\frac{\mu(1-\tilde{p})^2T}{S} \right)  +  T\frac{(2\tilde{p}-\tilde{p}^2)(G + 2d)(1+10\lambda)^2BK^2+ 6dK+8d}{2SBK^2}~,
\end{align}
 where the second inequality uses $\mathcal{L}(\bm{\theta}^{(0)}) - \mathcal{L}^* \leq  1 + 90\lambda d$ and   $1+x\leq e^x$ for all real $x$.

\end{proof}

 \subsection{Proof of Lemma \ref{lem:proof_thm_utility_R1_QNN}}\label{Append:subsec:R1_QNN_lem}
 \begin{proof}[Proof of Lemma \ref{lem:proof_thm_utility_R1_QNN}]
 	As shown in Theorem \ref{thm:noise_QNN_gaussian_formal}, the explicit formula of the estimated gradient is 
 	\begin{align}
 		\nabla_j \bar{\mathcal{L}}(\bm{\theta}^{(t)}) =\frac{1}{B}\sum_{i=1}^B (1-\tilde{p})^2\nabla_j{\mathcal{L}}_i(\bm{\theta}^{(t)}) + C_{j,1}^{(i,t)} + C_{j,2}^{(i,t)}\xi_{i}^{(t)} + C_{j,3}^{(i,t)}\xi_{i}^{(t,j)}+\xi_{i}^{(t)}\xi^{(t,j)}~.
 	\end{align}

By using the above result, we obtain
 \begingroup
 \allowdisplaybreaks
 \begin{align}\label{eqn:thm_emr_QNN_3_1}
  &	\frac{1}{2S}\sum_{j=1}^d  \mathbb{E}_{\xi_i^{(t)},\xi_i^{(t,j)}} \left[ \left(\nabla_j \bar{\mathcal{L}}(\bm{\theta}^{(t)})   \right)^2 \right] \nonumber\\
  \leq & \frac{(1-\tilde{p})^4 }{2S} \|  \nabla{\mathcal{L}}(\bm{\theta}^{(t)}) \|^2 + \frac{(1-\tilde{p})^2 }{2SB}\sum_{j=1}^d \nabla_j{\mathcal{L}}(\bm{\theta}^{(t)})\left(\sum_{i=1}^B C_{j,1}^{(i,t)} \right) + \frac{(1-\tilde{p})^2 }{SB}\sum_{j=1}^d \nabla_j{\mathcal{L}}(\bm{\theta}^{(t)}) \sum_{i=1}^B\mathbb{E}_{\xi_i^{(t)}}[\xi_i^{(t)}]\nonumber\\
  & + \frac{(1-\tilde{p})^2 }{SB}\sum_{j=1}^d\nabla_j{\mathcal{L}}(\bm{\theta}^{(t)}) \sum_{i=1}^B\mathbb{E}_{\xi_i^{(t,j)}}[\xi_i^{(t,j)}] + \frac{(1-\tilde{p})^2 }{SB}\sum_{j=1}^d \nabla_j{\mathcal{L}}(\bm{\theta}^{(t)}) \sum_{i=1}^B\mathbb{E}_{\xi_i^{(t)}\xi_i^{(t,j)}}[\xi_i^{(t)}\xi_i^{(t,j)}] \nonumber\\
  & + \frac{d}{2SB^2}\left(\sum_{i=1}^B C_{j,1}^{(i,t)} \right)^2 + \frac{1}{2S}\sum_{j=1}^d\mathbb{E}_{\xi_i^{(t)}}[\xi_i^{(t)}] + \frac{1}{2S}\sum_{j=1}^d\mathbb{E}_{\xi_i^{(t,j)}}[\xi_i^{(t,j)}]+ \frac{1}{2S}\sum_{j=1}^d\mathbb{E}_{\xi_i^{(t)},\xi_i^{(t,j)}}[\xi_i^{(t)}\xi_i^{(t,j)}] \nonumber\\
  & + \frac{1}{2SB^2}\sum_{j=1}^d\sum_{i=1}^B\mathbb{E}_{\xi_i^{(t)}}[(\xi_i^{(t)})^2] + \frac{1}{SB^2}\sum_{j=1}^d\sum_{i=1}^B \left(\mathbb{E}_{\xi_i^{(t)},\xi_i^{(t,j)}}[\xi_i^{(t)}\xi_i^{(t,j)}] + \mathbb{E}_{\xi_i^{(t)},\xi_i^{(t,j)}}[(\xi_i^{(t)})^2\xi_i^{(t,j)}]\right) \nonumber\\
  & + \frac{1}{2SB^2}\sum_{j=1}^d\sum_{i=1}^B\mathbb{E}_{\xi_i^{(t,j)}}[(\xi_i^{(t,j)})^2] + \frac{1}{SB^2}\sum_{j=1}^d\sum_{i=1}^B\mathbb{E}_{\xi_i^{(t)},\xi_i^{(t,j)}}[\xi_i^{(t)}(\xi_i^{(t,j)})^2] +\nonumber\\
  & + \frac{1}{2SB^2}\sum_{j=1}^d\sum_{i=1}^B\mathbb{E}_{\xi_i^{(t)}\xi_i^{(t,j)}} [(\xi_i^{(t)})^2(\xi_i^{(t,j)})^2] \nonumber\\
  \leq & \frac{(1-\tilde{p})^4 }{2S} \|  \nabla {\mathcal{L}}(\bm{\theta}^{(t)}) \|^2 + \frac{(1-\tilde{p})^2 G}{2S}\max_{i,j} C_{j,1}^{(i,t)} + \frac{d}{2S}\max_{i,j} \left( C_{j,1}^{(i,t)}\right)^2 \nonumber\\
  & + \frac{dC_{j,4,\max}^{(t)} }{2SB} + \frac{dC_{j,5,\max}^{(t,j)} }{2SB}+ \frac{dC_{j,4,\max}^{(t)}C_{j,5,\max}^{(t,j)} }{2SB}~.
 \end{align}
 \endgroup
 The first and second inequalities uses  $C_{j,2}^{(i,t)}\leq 1$,  $C_{j,3}^{(i,t)}\leq 1$, $\mathbb{E}[\xi^{(t)}_i]=0$, $\mathbb{E}[\xi^{(t,j)}_i]=0$, and  $- G/d \leq \nabla_j \mathcal{L}(\bm{\theta}^{(t)}) \leq G/d$ supported by $G$-Lipschitz property. The term $C_{j,4,\max}^{(t)}$ refers to $C_{j,4,\max}^{(t)} = \max_i C_{j,4}^{(i,t)}$. Similarly, the term $C_{j,5,\max}^{(t,j)}$ refers to $C_{j,5,\max}^{(t,j)} = \max_i   C_{j,5}^{(i,t)}$.
 
 Since Theorem \ref{thm:noise_QNN_gaussian_formal} indicates that 
  \[C_{j,4,\max}^{(t)} \leq  \frac{(1-\tilde{p})\left(1-2\tilde{p}\frac{\Tr(\Pi)}{D}\right)}{K}  +\tilde{p}\frac{\Tr(\Pi)}{DK} \leq   \frac{2}{K}~,\] 
    and 
   \[C_{j,5,\max}^{(t,j)} \leq \frac{(1-\tilde{p})\left(1-2\tilde{p}\frac{\Tr(\Pi)}{D}\right)(\hat{Y}_i^{(t,+_j)}+\hat{Y}_i^{(t,-_j)})+2\tilde{p}\frac{\Tr(\Pi)}{D}}{K}  \leq \frac{4}{K}~,\] we obtain  
\begin{align}
	& \frac{1}{2S}\sum_{j=1}^d  \mathbb{E}_{\xi_i^{(t)},\xi_i^{(t,j)}} \left[ \left(\nabla_j \bar{\mathcal{L}}(\bm{\theta}^{(t)})   \right)^2 \right] \nonumber\\
	\leq & \frac{(1-\tilde{p})^4 }{2S} \|  \nabla {\mathcal{L}}(\bm{\theta}^{(t)}) \|^2 + \frac{(1-\tilde{p})^2 G}{2S}\max_{i,j} C_{j,1}^{(i,t)} + \frac{d}{2S}\max_{i,j} \left( C_{j,1}^{(i,t)}\right)^2 + \frac{6dK+8d}{2SBK^2}~. 
\end{align}	
 \end{proof}

\section{Proof of Lemma \ref{lem:DP_QNN01}}\label{Appendix:subsec_DP_QNN_secv}

As shown in Theorem \ref{thm:noise_QNN_gaussian_formal}, the estimated gradient is center around the analytic gradients and is perturbed by the random noise $\bm{\varsigma}^{(t,j)}_i$ that follows the certain distribution. This behavior resembles a class of differentially private (DP) learning algorithm  \cite{dwork2014algorithmic}, where a certain type of noise is attached to  the gradients to achieve the privacy and utility guarantees. Driven by the similarity between noisy QNN and DP models, here we investigate whether noisy QNN can be treated as a DP learning model.

The proof of Lemma \ref{lem:DP_QNN01} leverages the composition property of DP model as summarized below. 
\begin{prop}[Composition property, \cite{dwork2006calibrating}] \label{prop:DP_post_process}
	Suppose that a mechanism $\mathcal{M}$ consists of a sequence of adaptive $(\epsilon,\delta)$-differentially private mechanisms $\mathcal{M}_{1},...,\mathcal{M}_{k}$, where $\epsilon,\delta\geq 0$ $\mathcal{M}_{i}:\prod_{j=1}^{i-1}\mathcal{R}_j\times \mathbb{R}^D \rightarrow \mathcal{R}_i$. Then the mechanism $\mathcal{M}$ satisfies $(\epsilon', k\delta+\delta')$-differentially private with $\delta'\geq 0$ and 
\[\epsilon' = \sqrt{2k\ln(1/\delta')\epsilon} + k\epsilon(e^{\epsilon}-1)~.\]  
\end{prop} 

\begin{proof}[Proof of Lemma \ref{lem:DP_QNN01}]
Recall that, for noisy QNN, the estimated gradient of $j$-th parameter at $t$-th is 
\begin{equation}\label{eqn:DP_QNN_1}
	\nabla_j\bar{\mathcal{L}}_i(\bm{\theta}^{(t)}) = (\bar{Y}_i^{(t)} - {Y}_i)\left(\bar{Y}_i^{(t,+_j)} - \bar{Y}_i^{(t,-_j)}\right) + \lambda \bm{\theta}_j^{(t)}~.
\end{equation}
The composition property of DP as shown  in Proposition \ref{prop:DP_post_process} indicates that, if the mechanism $\mathcal{M}(\bm{\theta}_j ^{(t)},\mathcal{B}_i)$ that corresponds to the quantum circuits as shown in Fig.~\ref{fig:QNN} (b), which is used to output $\bar{Y}_i^{(t)}$ and $\bar{Y}_i^{(t,\pm_j)}$, satisfies DP property, then QNN with noisy gates also achieves the DP promise. Alternatively, to guarantee the DP property of QNN, we should prove that the  random mechanism $\mathcal{M}(\bm{\theta}_j ^{(t)},\mathcal{B}_i)$ is a DP model. Without loss of generality, here we focus on the setting with $B=1$ and $\mathcal{B}_i=\bm{z}$, since the privacy keeps unchanged when we vary $B$ from $1$ to $N$.

As explained in Theorem \ref{thm:noise_QNN_gaussian_formal}, the randomness of the mechanism $\mathcal{M}(\bm{\theta}_j ^{(t)},\bm{z})$ comes from the gate noise and finite number of measurements.  For $K$ quantum measurements, the possible values of sample mean $\bar{Y}_i^{(t)}$ is discrete, i.e., $\bar{Y}_i^{(t)}\in\{0, 1/K,...,1\}$. By employing the properties of sample mean and Bernoulli random variables, the distribution of $\bar{Y}_i^{(t)}$ follows
\begin{equation}\label{Eqn:Appendix_DP_QNN_grad}
	\Pr(\bar{Y}_i^{(t)}=y)=  \begin{psmallmatrix}K \\ Ky  \end{psmallmatrix}q^{Ky}(1-q)^{K-Ky}~,
\end{equation}
where $q=(1-\tilde{p})\Tr(\hat{Y}_i^{(t)})+\tilde{p}\Tr(\Pi)/D$. 
  
In conjunction with Eqn.~(\ref{Eqn:Appendix_DP_QNN_grad}) and the definition of DP as formulated in Definition \ref{def:CDP}, the random algorithm $\mathcal{M}(\bm{\theta}^{(t)}_j,\bm{z})$ is DP  if the following relation is satisfied, i.e.,
\begin{equation}\label{Eqn:Appendix_DP_QNN_grad_2}
	\frac{\Pr(\bar{Y}_i^{(t)}=y)}{\Pr(\bar{Y}_i^{'(t)}=y)} \leq e^{\epsilon'}~,
\end{equation}
where $\bar{Y}_i^{'(t)}$ refers to the sample mean of QNN given the tunable parameters $\bm{\theta}^{(t)}$ and the neighborhood dataset $\bm{z}'$. Combining Eqn.~(\ref{Eqn:Appendix_DP_QNN_grad}) and Eqn.~(\ref{Eqn:Appendix_DP_QNN_grad_2}), we obtain  
\begin{align}\label{Eqn:Appendix_DP_QNN_grad_3}
	& \frac{\Pr(\bar{Y}_i^{(t)}=y)}{\Pr(\bar{Y}_i^{'(t)}=y)}= \frac{\begin{psmallmatrix}K \\ Ky  \end{psmallmatrix}q^{Ky}(1-q)^{K-Ky}}{\begin{psmallmatrix}K \\ Ky  \end{psmallmatrix}{q'}^{Ky}(1-{q'})^{K-Ky}} \leq   \frac{q}{(q'(1-q'))^{K}}~,
\end{align}
where $q'=(1-\tilde{p})\Tr(\hat{Y}_i^{'(t)})+\tilde{p}\Tr(\Pi)/D$, the inequality uses the facts $q^{Ky}(1-q)^{K-Ky}\leq q$ and ${q'}^{Ky}(1-{q'})^{K-Ky}\geq (q'(1-q'))^{K}$.

By replacing $q$ and $q'$ with their explicit expressions, Eqn.~(\ref{Eqn:Appendix_DP_QNN_grad_3}) can be further simplified as 
\begin{align}
	\frac{\Pr(\bar{Y}_i^{(t)}=y)}{\Pr(\bar{Y}_i^{'(t)}=y)}\leq \frac{(1-\tilde{p})+\tilde{p}\frac{\Tr(\Pi)}{D}}{\left(\tilde{p}(1-\tilde{p})(1-\frac{\Tr(\Pi)}{D})\right)^K}~,
\end{align}  
where the nominator employs $\Tr(\hat{Y}_i^{(t)})\leq 1$ and the denominator uses the fact $\tilde{p}\Tr(\Pi)/D \leq q'\leq (1-\tilde{p})+\tilde{p} \Tr(\Pi)/D$.

The result achieved in Eqn.~(\ref{Eqn:Appendix_DP_QNN_grad_3}) indicates that the mechanism $\mathcal{M}(\bm{\theta}^{(t)}_j,\bm{z})$ is a DP model, where
\begin{equation}
	\frac{(1-\tilde{p})+\tilde{p}\frac{\Tr(\Pi)}{D}}{\left(\tilde{p}(1-\tilde{p})(1-\frac{\Tr(\Pi)}{D})\right)^K} = e^{\epsilon'} \Leftrightarrow \epsilon' = \ln\left(\frac{(1-\tilde{p})+\tilde{p}\frac{\Tr(\Pi)}{D}}{\left(\tilde{p}(1-\tilde{p})(1-\frac{\Tr(\Pi)}{D})\right)^K}\right) ~.
\end{equation}

We then use the DP property of the mechanism $\mathcal{M}$ to derive the privacy parameter of QNN at the $t$-th iteration. By leveraging Proposition \ref{prop:DP_post_process}, the privacy parameters $(\epsilon'', \delta'')$ of QNN to generate the estimated gradient of the $j$-th parameter $\nabla_j \mathcal{L}_i$ is 
\begin{equation}\label{Eqn:Appendix_DP_QNN_grad_4}
	\epsilon'' = \sqrt{6\ln(1/\delta'')\epsilon'}+ 3\epsilon'(e^{\epsilon'}-1)~.
\end{equation} 

Since the $d$ trainable parameters of $\nabla  \mathcal{L}_i$ are independent with each other, the definition of DP requests that, given two neighborhood input datasets $\bm{z}$ and $\bm{z}'$, the following relation should be satisfied at the $t$-th iteration,
\begin{equation}\label{eqn:thm-DP-conc}
	\prod_{j=1}^d \max_{r_j} \frac{\Pr(\mathcal{M}(\bm{\theta}_j^{(t)}, \bm{z})=  r_j)}{\Pr(\mathcal{M}(\bm{\theta}^{(t)}_j,\bm{z}')=  r_j)} \leq e^{\epsilon'''} + \delta''' ~.
\end{equation}

In conjunction with Eqn.~(\ref{Eqn:Appendix_DP_QNN_grad_4}) and  Eqn.~(\ref{eqn:thm-DP-conc}), we obtain
\begin{equation}
	\epsilon''' = d\epsilon'',~\text{and}~,\delta'''=d\delta''~.
\end{equation}

Since the mechanism of QNN that is used  to generate $\nabla \mathcal{L}_i$ at the $t$-th iteration satisfies the $(\epsilon''',\delta''')$-DP property, we can utilize Proposition \ref{prop:DP_post_process} again to show QNN with $T$ iterations is also an $(\epsilon, \delta)$-DP model, i.e.,
\begin{align}
	\epsilon & = \sqrt{2T\ln(1/\bar{\delta})d\epsilon''}+ Td\epsilon''(e^{d\epsilon''}-1) \nonumber\\
	& = \tilde{O}\left(\sqrt{Td}+Td \left(\frac{(1-\tilde{p})+\tilde{p}\frac{\Tr(\Pi)}{D}}{\left(\tilde{p}(1-\tilde{p})(1-\frac{\Tr(\Pi)}{D})\right)^K}\right)^d - Td\right)~.
\end{align}

\end{proof}

\section{Proof of Theorem \ref{thm:DP_QNNQAE_inform}}\label{Appendix:proof_thm2_QNN_learnability}

The key ingredients to achieve Theorems \ref{thm:DP_QNNQAE_inform} are classical and quantum differentially private (DP) learning techniques   \cite{chaudhuri2011differentially,du2020quantum,dwork2014algorithmic,zhou2017differential}, and quantum $\PAC$ and  $\QSQ$ learning models \cite{arunachalam2017guest,arunachalam2020quantum}. The intuition to employ DP is as follows. The behavior of QNN with gate noise resembles DP learning, where a certain type of  noise is injected into the learning model.  Moreover, a recent study \cite{arunachalam2020quantum} proved that if a learning problem is quantum $\PAC$ learnable, then it is also quantum privately $\PAC$ ($\PPAC$) learnable.  Such an observation implies that if QNN with gate noise belongs to the DP learning model, then we can conclude the same learnability between noiseless QNN and QNN with gate noise.

To incorporate the achieved result of QNN with other quantum learning theory conclusions, the quantum examples discussed below concentrate on a specific type as formulated in Definition \ref{def:Q_exmaple}, which is broadly employed in quantum $\PAC$ learning and quantum statistical query ($\QSQ$) learning. Note that the quantum encoding circuit $U_{\bm{x}}$ can efficiently prepare such quantum examples, as explained in Appendix \ref{appen:QNN}.

\begin{definition}[Quantum example]\label{def:Q_exmaple}
Let $c^*:\{0,1\}^N\rightarrow{0,1}$ be an unknown concept sampled from a known concept class $\mathcal{C}$. Denote the labeled examples as $(\bm{x}, c^*(\bm{x}))$, where $\bm{x}$ is  drawn from some unknown distribution $\mathcal{D}$. The quantum example is defined as $\ket{\psi_{c^*}} = \sum_{\bm{x}\in\{0,1\}^N} \sqrt{\mathcal{D}(\bm{x})}\ket{\bm{x}}\ket{c^*({\bm{x}})}$. 
\end{definition}

\begin{proof}[Proof of Theorem \ref{thm:DP_QNNQAE_inform}]
	Given access to quantum examples $\ket{\psi_{c^*}}$ as formulated in Definition \ref{def:Q_exmaple}, we can leverage the results of quantum learning theory \cite{arunachalam2017guest,arunachalam2020quantum} to exploit the learnability of noiseless QNN and QNN with noisy gates. In particular, the two studies \cite{arunachalam2017guest,arunachalam2020quantum} proved $\text{quantum} \PAC = \PAC$ and $\text{quantum} \PPAC = \PPAC$. Since a well known classical result \cite{kasiviswanathan2011can} is  $\PAC=\PPAC$, we obtain the following relationship in terms of sample complexity, i.e., 
	\begin{align}\label{eqn:prof_learn_QNN_QAE_1}
		\text{quantum} \PAC = \text{quantum} \PPAC = \PAC = \PPAC~. 
	\end{align} 
Eqn.~(\ref{eqn:prof_learn_QNN_QAE_1}) indicates that the learnable concept classes for non-private learning model and the DP learning model are same. Consequently,   if a concept is $\PAC$ learnable by a QNN, then such a concept is also  $\PAC$  learnable by a DP learning model, i.e., QNN with noisy gates, where its  DP property has been proved in Lemma \ref{lem:DP_QNN01}. 
\end{proof}

\section{Proof of Theorem \ref{thm:QNN_QAE_QSQ_info}}\label{Appendix:proof_thm3_QNN_learnability}
Theorem \ref{thm:QNN_QAE_QSQ_info} quantifies  the required query complexity of QNN with noisy gates to simulate one query of $\QSQ$ model. The definition of quantum statistical query ($\QSQ$) learning model and its relevant  theoretical results \cite{arunachalam2020quantum} as shown below.
\begin{definition}[$\QSQ$]\label{def:QSQ}
	Let $\mathcal{C}\subseteq \{c:\{0,1\}^N\}\rightarrow \{0, 1\}$ be a concept class and $\mathcal{D}:\{0, 1\}^N \rightarrow \{0, 1\}$ be a distribution. A quantum statistical query oracle $\Qstat(\Pi, \tau)$ for some $c^*\in \mathcal{C}$ receives as inputs a tolerance $\tau\geq 0$ and an observable  $\mathbb{M}\in (\mathbb{C}^2)^{\otimes N+1}\times (\mathbb{C}^2)^{\otimes N+1}$, and outputs a number $\alpha$ satisfying \[|\alpha - \braket{\psi_{c^*}|\mathbb{M}| \psi_{c^*}}|\leq \tau~,\] where $\psi_{c^*}=\sum_{\bm{x}\in\{0,1\}^N}\sqrt{\mathcal{D}(\bm{x})}\ket{\bm{x}, c^*({\bm{x}})}$ refers to the quantum example. 
\end{definition} 

\begin{definition}[$\varepsilon$-learning]
	Let $\mathcal{C}\subseteq \{c:\{0,1\}^N\}\rightarrow \{0, 1\}$ be a concept class and $\mathcal{D}:\{0, 1\}^N \rightarrow \{0, 1\}$ be a distribution. We say that $\mathcal{C}$ can be $\varepsilon$-learned  in the $\QSQ$ model with $Q$ queries, if there is an algorithm $\mathcal{A}$ such that for every $c^*\in\mathcal{C}$, $\mathcal{A}$ makes at most $Q$ $\Qstat$ queries and outputs a hypothesis $h$ satisfying $\Pr_{\bm{x}\sim \mathcal{D}}[h(\bm{x})\neq c^*(\bm{x})]\leq \varepsilon$.
\end{definition}

The key technique to achieve Theorem \ref{thm:QNN_QAE_QSQ_info} is the concentration inequality, which bounds the deviation of a random variable that corresponds to the output of QNN from a certain number.  In particular, with treating the output $\alpha$ in $\QSQ$ model as the sample mean of QNN, the relation $\alpha - \braket{\psi_{c^*}|\mathbb{M}| \psi_{c^*}}|\leq \tau$ evaluates how is the probability when the distance between the sample mean $\alpha$ and its expectation $\braket{\psi_{c^*}|\mathbb{M}\psi_{c^*}}$ is  within $\tau$. Such a question can be effectively answered by using concentration inequality.

\begin{lem}[Modified from Lemma 4.2, 4.3, and 4.5 in \cite{arunachalam2020quantum}]\label{lem:QSQ_res}
Let $\mathcal{C}$ be the concept class of parities, $k$-juntas, or poly($n$)-sized DNFs (Disjunctive Normal Forms), then there exists a $poly(n)$ queries $\QSQ$ algorithm with tolerance $\tau=\tilde{O}(\varepsilon)$ that $\varepsilon$-learns  $\mathcal{C}$ under the uniform distribution. All of these concepts are computationally hard for $\SQ$ models.  
\end{lem}

\begin{proof}[Proof of Theorem \ref{thm:QNN_QAE_QSQ_info}]
Following the notations used in  Definitions \ref{def:Q_exmaple} and \ref{def:QSQ}, supposed that the encoding circuits $U_{\bm{x}}$ prepares the quantum example $\ket{\psi_{c^*}}$ and the trainable unitary $U(\bm{\theta})$ is identity $\mathbb{I}_{2^{\otimes N+1}}$. Then,  with applying the observable $\mathbb{M}$ to the generated state of QNN, the expectation value of quantum measurements under the depolarization noise setting $\mathcal{N}_{\tilde{p}}$ follows $\tilde{\nu}= (1-\tilde{p})\nu + \frac{1}{2^{N+1}}$ with $\nu=\braket{\psi_{c^*}|\mathbb{M}|\psi_{c^*}}$, supported by Lemma \ref{lem:equi_dep}.  The measurement outcome $V_k$ is a random variable that satisfies $V_k\sim \Ber(\tilde{\nu})$.

 By the Chernoff-Hoeffding bound for real-valued variables, we obtain the relation between the sample mean $\frac{1}{K}\sum_{k=1}^K V_k$ with $K$ measurements and the target result $\tilde{\nu}$, i.e., 
 \begin{equation}
 	\Pr\left(\left|\frac{1}{K}\sum_{i=1}^K V_k - \tilde{\nu} \right| \geq 
 \frac{\delta}{2} \right) \leq 2 \exp(-{\delta^2 K}/2)~.
 \end{equation}
 
Moreover, the distance between the target result $\nu$ and the shifted expectation values $\tilde{\nu}$ follows 
 \begin{equation}
 	|\nu- \tilde{Y}|\leq \tilde{p}\nu + \frac{\Tr(\mathbb{M})}{2^{N+1}}~.  
 \end{equation}

 In conjunction with the above two equations, we obtain, with probability at least $1- 2 \exp(-{\delta^2 n}/2)$
 \begin{equation}
 	\left|\frac{1}{K}\sum_{k=1}^K V_k - \nu \right| = 	\left|\frac{1}{K}\sum_{k=1}^K V_k -\tilde{\nu} + \tilde{\nu} - \nu \right|  \leq  \tilde{p}\nu + \frac{\Tr(\mathbb{M})}{2^{N+1}} +  \frac{\delta}{2} ~.
 \end{equation}
 
 Note that, to guarantee the term $\tilde{p}\nu + \frac{\Tr(\mathbb{M})}{2^{N+1}} +  \frac{\delta}{2}$ is upper bounded by $\tau$, the parameter $\tilde{p}$ should satisfy 
 \begin{equation}\label{eqn:thm_QNN_QSQ_1}
 	\tilde{p} \leq \frac{\tau - \frac{\delta}{2}-\frac{\Tr(\mathbb{M})}{2^{N+1}}}{\nu}~.
 \end{equation} 
 
Under the assumption of  Eqn.~(\ref{eqn:thm_QNN_QSQ_1}),  with setting $\delta \leq 2(\tau - \tilde{p}\nu -  \frac{\Tr(\mathbb{M})}{2^{N+1}})$, QNN simulates  the  $\QSQ$ model as formulated in Definition \ref{def:QSQ}, i.e.,  \[\left|\frac{1}{K}\sum_{k=1}^K V_k - \nu \right| \leq  \tau .\] Under this setting, the relation between the number of measurements $K$ and the successful probability $b$ obeys  
 \begin{equation}
 	\Pr\left(\left|\frac{1}{K}\sum_{k=1}^K V_k - \tilde{\nu} \right| \geq 
\left(\tau - \tilde{p}\nu -  \frac{\Tr(\mathbb{M})}{2^{N+1}}\right) \right) \leq 2 \exp\left(-2{\left(\tau - \tilde{p}\nu -  \frac{\Tr(\mathbb{M})}{2^{N+1}}\right)^2 K} \right)=b~.
 \end{equation}
After simplification, we conclude that, when $\tilde{p} \leq \frac{\tau - \frac{\delta}{2}-\frac{\Tr(\mathbb{M})}{2^{N+1}}}{\nu}$, with the successful probability at least $1-b$,  the required number of measurements  to attain $\left|\frac{1}{K}\sum_{k=1}^K V_k - \nu \right| \leq  \tau$ is
\begin{equation}
	K =  \frac{\ln\left(\frac{2}{b}\right)}{2 \left(\tau - \tilde{p}\nu -  \frac{\Tr(\mathbb{M})}{2^{N+1}}\right)^2}~.
\end{equation}

\end{proof}

 \section{Generalization the results to more general quantum channels} \label{Appendix:sec_more_general_channel}
In this section, we generalize the achieved results in main text from the  depolarization channel to a more general channel $\mathcal{E}_{p_1}$, i.e., 
\begin{equation}
	\mathcal{E}_{p_1}(\rho)=(1-p_1)\rho + p_2 \kappa + p_3\mathbb{I}_{D}/D ~,
\end{equation}
where $\rho, \kappa \in \mathbb{C}^{D\times D}$, $\kappa$ is a mixed state that can either be correlated or uncorrelated with $\rho$, and $p_2+p_3=p_1$ with $p_1,p_2\geq 0$ and $p_3>0$. It is worth noting that the quantum channel $\mathcal{E}_{p_1}$ is sufficiently universal, which covers most Pauli channels associated with the depolarization channel \cite{nielsen2010quantum,sharma2020noise}.  

The outline of this section is as follows. In Subsection \ref{Appendix:subsec_utility_QNN_general_channel}, we discuss the utility bounds of QNN under the general channel setting. Then, in Subsection \ref{Appendix:subsec_DP_QNN_general_channel}, we analyze the DP property of QNN when it is  perturbed by the general channel. Last, in Subsection \ref{Appendix:subsec_Learn_QNN_general_channel}, we quantify the learnability of QNN  under  the general channel setting from the perspective of sample complexity.   

\subsection{Utility bounds of QNN}\label{Appendix:subsec_utility_QNN_general_channel}
Analogous to the depolarization channel setting, we first simplify the noisy QNN model to ease analysis. Specifically, after applying $\mathcal{E}_{p_1}$ to each circuit depth, the generated state follows
\begin{align}\label{eqn:append_general_channel_1}
	& \mathcal{E}_{p_1}({U}_L(\bm{\theta})...{U}_2(\bm{\theta})\mathcal{E}_{p_1}({U}_1(\bm{\theta})\rho {U}_1(\bm{\theta})^{\dagger}){U}_2(\bm{\theta})^{\dagger}...{U}_L(\bm{\theta})^{\dagger}) \nonumber\\
	 = & (1-p_1)^{L_Q}\left({U}(\bm{\theta}) U_{\bm{x}}\right)\rho \left({U}(\bm{\theta}) U_{\bm{x}}\right)^{\dagger} + p_2'\kappa+ p_3^{L_Q}\frac{\mathbb{I}_D}{D}~, 
\end{align}    
  where  $(1-p_1)^{L_Q}+ p_2'+p_3^{L_Q}=1$, and $\kappa$ is a mixed state that can either be correlated or uncorrelated with $\left({U}(\bm{\theta}) U_{\bm{x}}\right) \rho \left({U}(\bm{\theta}) U_{\bm{x}}\right)^{\dagger}$. Without confusion, we set $\tilde{p}=1-(1-p_1)^{L_Q}$.   

We now employ the simplified model, i.e., the right hand side of Eqn.~(\ref{eqn:append_general_channel_1}), to establish the relation between the estimated gradients $\nabla_j \bar{\mathcal{L}}_i(\bm{\theta}^{(t)})$ and the analytic gradients $\nabla_j {\mathcal{L}}_i(\bm{\theta}^{(t)})$. Recall that \[\nabla_j \bar{\mathcal{L}}_i(\bm{\theta}^{(t)}) =  (\bar{Y}_i^{(t)} - {Y}_i)\left(\bar{Y}_i^{(t,+_j)} - \bar{Y}_i^{(t,-_j)}\right) + \lambda \bm{\theta}_j^{(t)}~,\] where $\bar{Y}_i^{(t)}=\sum_{k=1}^K V_k^{(t)}/K$ and $\bar{Y}_i^{(t,\pm_j)}=\sum_{k=1}^K V_k^{(t,\pm_j)}/K$ refer to the  sample means when feeding $\bm{\theta}^{(t)}$ and $\bm{\theta}^{(t,\pm_j)}$ into the trainable circuit.  As with depolarization channel, the sample mean $\bar{Y}_i^{(t)}$ or $\bar{Y}_i^{(t,\pm_j)}$ is a random variable follows certain distribution. In particular, following the notations used in Theorem \ref{thm:noise_QNN_gaussian_formal}, the mean and variance of $\bar{Y}_i^{(t)}$  follows
\[
\begin{cases}
   \nu^{(t)} = (1-\tilde{p})\hat{Y}_i^{(t)} + p_2'\Tr(\Pi\kappa^{(t)})+ \frac{p_3^{L_Q}}{2} ~,  \\
  \sigma^{(t)}= -\frac{\left( (1-\tilde{p})\hat{Y}_i^{(t)} + p_2'\Tr(\Pi\kappa^{(t)}) \right)^2}{K} + \frac{(1-p_3^{L_Q})\left( (1-\tilde{p})\hat{Y}_i^{(t)} + p_2'\Tr(\Pi\kappa^{(t)}) \right)}{K} + \frac{p_3^{L_Q}}{2}-\frac{(p_3^{L_Q})^2}{4}. 
  \end{cases}
\]
Similarly, the mean and variance of $\bar{Y}_i^{(t,\pm_j)}$  follows
\[
\begin{cases}
   \nu^{(t,\pm_j)} = (1-\tilde{p})\hat{Y}_i^{(t,\pm_j)} + p_2'\Tr(\Pi\kappa^{(t,\pm_j)})+ \frac{p_3^{L_Q}}{2} ~,  \\
  \sigma^{(t,\pm_j)}= -\frac{\left( (1-\tilde{p})\hat{Y}_i^{(t,\pm_j)} + p_2'\Tr(\Pi\kappa^{(t,\pm_j)}) \right)^2}{K} + \frac{(1-p_3^{L_Q})\left( (1-\tilde{p})\hat{Y}_i^{(t,\pm_j)} + p_2'\Tr(\Pi\kappa^{(t,\pm_j)}) \right)}{K} + \frac{p_3^{L_Q}}{2}-\frac{(p_3^{L_Q})^2}{4}. 
  \end{cases}
\]
By expanding the sample means using their explicit forms as shown above, we obtain the relation between the estimated and analytic gradients, i.e.,
\begin{align}
	\nabla_j \bar{\mathcal{L}}_i(\bm{\theta}^{(t)}) & = (1-\tilde{p})^2 \nabla_j {\mathcal{L}}_i(\bm{\theta}^{(t)}) + C_{j,1}^{(i,t)} + \bm{\varsigma}_i^{(t,j)}  ~, 
\end{align}
where $\bm{\varsigma}_i^{t,j}=C_{j,2}^{(i,t)}\xi_i^{(t)}+C_{j,2}^{(i,t)}\xi_i^{(t,j)}+ \xi_i^{(t)}\xi_i^{(t,j)}$, and two random variables $\xi_i^{(t)}$ and $\xi_i^{(t)}$ have zero means and their variances are $C_{j,4}^{(i,t)}$ and $C_{j,5}^{(i,t)}$, respectively. The explicit formula of the five parameters $\{C_{j,a}^{(i,t)}\}_{a=1}^t$ is 
\[
\begin{cases}
	C_{j,1}^{(i,t)} = & \left(p_2'\Tr(\Pi\kappa^{(t)}) + \frac{p_3^{L_Q}}{2}-\tilde{p}Y_i\right)(1-\tilde{p})(\hat{Y}_i^{(t,+_j)}-\hat{Y}_i^{(t,-_j)}) \\
	  & + p_2'(1-\tilde{p})(\hat{Y}_i^{(t)}-Y_i)(\Tr(\Pi\kappa^{(t,+_j)}) - \Tr(\Pi\kappa^{(t,-_j)})) \nonumber\\
	  & +  \left(p_2'\Tr(\Pi\kappa^{(t)}) + \frac{p_3^{L_Q}}{2}-\tilde{p}Y_i\right)(\Tr(\Pi\kappa^{(t,+_j)}) - \Tr(\Pi\kappa^{(t,-_j)}))+(1-(1-\tilde{p})^2)\lambda \bm{\theta}_j^{(t)} ~, \\
	  C_{j,2}^{(i,t)} = &  \left((1-\tilde{p})(\hat{Y}_i^{(t,+_j)}-\hat{Y}_i^{(t,-_j)}) + p_2'(\Tr(\Pi\kappa^{(t,+_j)}) - \Tr(\Pi\kappa^{(t,-_j)})) \right)~,\\
	  C_{j,3}^{(i,t)} = &  \left((1-\tilde{p})(\hat{Y}_i^{(t)}-Y_i)+ \left(p_2'\Tr(\Pi\kappa^{(t)}) + \frac{p_3^{L_Q}}{2}-\tilde{p}Y_i\right)  \right)~,\\
	   C_{j,4}^{(i,t)} = & -\frac{\left( (1-\tilde{p})\hat{Y}_i^{(t)} + p_2'\Tr(\Pi\kappa^{(t)}) \right)^2}{K} + \frac{(1-p_3^{L_Q})\left( (1-\tilde{p})\hat{Y}_i^{(t)} + p_2'\Tr(\Pi\kappa^{(t)}) \right)}{K} + \frac{p_3^{L_Q}}{2K}-\frac{(p_3^{L_Q})^2}{4K}~, \\
	     C_{j,5}^{(i,t)} = & -\frac{\left( (1-\tilde{p})\hat{Y}_i^{(t,+_j)} + p_2'\Tr(\Pi\kappa^{(t,+_j)}) \right)^2}{K} -\frac{\left( (1-\tilde{p})\hat{Y}_i^{(t,-_j)} + p_2'\Tr(\Pi\kappa^{(t,-_j)}) \right)^2}{K} \\
	     &+ \frac{(1-p_3^{L_Q})\left( (1-\tilde{p})(\hat{Y}_i^{(t,+_j)}-\hat{Y}_i^{(t,-_j)}) + p_2'(\Tr(\Pi\kappa^{(t,+_j)})-\Tr(\Pi\kappa^{(t,-_j)})) \right)}{K} +  \frac{p_3^{L_Q}}{K} -\frac{(p_3^{L_Q})^2}{2K}~.
\end{cases}
\]  
We next use the relation between the estimated and analytic gradients to separately quantify the utility bounds $R_1$ and $R_2$ of QNN under the noisy channel $\mathcal{E}_{p_1}$ setting.

\textbf{Utility bound $R_1$.} As with Eqn.(\ref{eqn:thm_emr_QNN_3}), with taking expectation over $\xi_i^{(t)}$ and $\xi_i^{(t,j)}$, we obtain 
\begin{align}\label{eqn:append_general_channel_2} 
	& \mathbb{E}_{\xi_i^{(t)},\xi_i^{(t,j)}}[\mathcal{L}(\bm{\theta}^{(t+1)}) -\mathcal{L}(\bm{\theta}^{(t)})] \nonumber\\
 \leq & -\frac{1}{S}(1-\tilde{p})^2 \|\nabla \mathcal{L}(\bm{\theta}^{(t)}) \|^2 + \frac{G}{2S}\left(\frac{1}{B}\sum_{i=1}^B C_{j,1}^{(i,t)} \right) +   \frac{1}{2S}\sum_{j=1}^d  \mathbb{E}_{\xi_i^{(t)},\xi_i^{(t,j)}} \left[ \left(\nabla_j \bar{\mathcal{L}}(\bm{\theta}^{(t)})   \right)^2 \right]~,
 \end{align}
where the inequality employs $\mathbb{E}[\xi^{(t)}_i]=0$, $\mathbb{E}[\xi^{(t,j)}_i]=0$, and $- G/d \leq \nabla_j \mathcal{L}(\bm{\theta}^{(t)}) \leq G/d$.
 
 For the term $\frac{1}{2S}\sum_{j=1}^d  \mathbb{E}_{\xi_i^{(t)},\xi_i^{(t,j)}} [ \left(\nabla_j \bar{\mathcal{L}}(\bm{\theta}^{(t)})   \right)^2 ]$ in the above equation, its upper bound satisfies 
 \begin{align}
 	\frac{1}{2S}\sum_{j=1}^d  \mathbb{E}_{\xi_i^{(t)},\xi_i^{(t,j)}} \left[ \left(\nabla_j \bar{\mathcal{L}}(\bm{\theta}^{(t)})   \right)^2 \right] 
 \leq    & \frac{(1-\tilde{p})^4 }{2S} \|  \nabla{\mathcal{L}}(\bm{\theta}^{(t)}) \|^2 + \frac{(1-\tilde{p})^2 G}{2SB}\sum_{i=1}^B C_1^{(i,t)}   \nonumber\\
  & + \frac{d}{2SB^2}\left(\sum_{i=1}^B C_1^{(i,t)} \right)^2  + d\frac{\sigma^{(t)}_{\max} +  \sigma^{(t,j)}_{\max} +  \sigma^{(t)}_{\max}\sigma^{(t,j)}_{\max}}{SB}~,
 \end{align}
where the first and second inequalities uses  $C_2^{(i,t)}\leq 2$,  $C_3^{(i,t)}\leq 2$, $\mathbb{E}[\xi^{(t)}_i]=0$, and $\mathbb{E}[\xi^{(t,j)}_i]=0$. The term $\sigma^{(t)}_{\max}$ refers to $\sigma^{(t)}_{\max} = \max_i \sigma^{(t)}_{i} \leq 3/K$. Similarly, the term $\sigma^{(t,j)}_{\max}$ refers to $\sigma^{(t,j)}_{\max} = \max_i \sigma^{(t,+_j)}_{i} + \sigma^{(t,-_j)}_{i} \leq 3/K$.

In conjunction with the above two equations, we achieve
\begin{align} \label{eqn:append_general_channel_3}
	& \mathbb{E}_{\xi_i^{(t)},\xi_i^{(t,j)}}[\mathcal{L}(\bm{\theta}^{(t+1)}) -\mathcal{L}(\bm{\theta}^{(t)})] \nonumber\\
	\leq & -\frac{1}{2S}(1-\tilde{p})^2 \|\nabla \mathcal{L}(\bm{\theta}^{(t)}) \|^2 + \frac{(2G+d)(5 + 3(1-(1-\tilde{p})^2)\lambda \pi)}{2S}+  \frac{6dK+9d}{SBK^2}~,   
\end{align}
where the inequality uses $C_{j,1}^{(i,t)}\leq 5 + 3(1-(1-\tilde{p})^2)\lambda \pi$.

After rewriting and taking induction, we have
\begin{equation}
\|\nabla \mathcal{L}(\bm{\theta}^{(t)}) \|^2 \leq  2S\frac{1+9\lambda d}{T(1-\tilde{p})^2}   +\frac{(2G+d)(5 + 3(1-(1-\tilde{p})^2)\lambda \pi)}{(1-\tilde{p})^2}  + \frac{12dK+18d}{(1-\tilde{p})^2BK^2}~.
\end{equation}	
With setting $T\rightarrow \infty$, we achieve the utility bound $R_1$, i.e.,
\begin{equation}
	R_1 \leq \tilde{O}\left(\frac{1}{(1-\tilde{p})^2}, d, \frac{1}{BK}\right)~.
\end{equation}

\textbf{Utility bound $R_2$.} With combining Eqn.~(\ref{eqn:append_general_channel_3}) and PL condition, we obtain
\begin{align}
	& \mathbb{E}_{\xi_i^{(t)},\xi_i^{(t,j)}}[\mathcal{L}(\bm{\theta}^{(t+1)}) -\mathcal{L}(\bm{\theta}^{(t)})] \nonumber\\
	\leq & -\frac{\mu(1-\tilde{p})^2}{S}(\mathcal{L}(\bm{\theta}^{(t)}) - \mathcal{L}^*) + \frac{(2G+d)(5 + 3(1-(1-\tilde{p})^2)\lambda \pi)}{2S}+  \frac{6dK+9d}{SBK^2}~. 
\end{align}
After rewriting and induction, we have
\begin{align}
	& \mathbb{E}_{\bm{\varsigma}^{(t)}}[\mathcal{L}(\bm{\theta}^{(T)})] - \mathcal{L}^* \leq 15\lambda d \exp\left(-\frac{\mu(1-\tilde{p})^2T}{S} \right)  +  T\frac{(2G+d)(5 + 3(1-(1-\tilde{p})^2)\lambda \pi)}{2S}+  T\frac{6dK+9d}{SBK^2}~.
\end{align}
With setting $T=O\left(\frac{S}{\mu(1-\tilde{p})^2}\ln\left( \frac{30\lambda dSBK^2}{(2G+d)(5 + 3(1-(1-\tilde{p})^2)\lambda \pi)BK^2 + 12dK + 18d} \right)\right)$,  the utility bound is
\begin{equation}
	R_2\leq O\left(\frac{1}{(1-\tilde{p})^2}, \frac{1}{SBK^2}, d \right)~.
\end{equation} 
\subsection{Differential privacy of QNN}\label{Appendix:subsec_DP_QNN_general_channel}
Analogous to  QNN with depolarization noise, the DP property of QNN perturbed by the cannel $\mathcal{E}_{p_1}$ is determined by the DP property of the mechanism $\mathcal{M}(\bm{\theta}_j,\bm{z})$ to output $Y_i^{(t)}$, supported by the composition property as shown in Proposition \ref{prop:DP_post_process}. 

As discussed in Subsection \ref{Appendix:subsec_utility_QNN_general_channel}, the distribution of sample mean $Y_i^{(t)}$ is similar to the depolarization case, where the only difference is that the values of mean and variance of the random variable are varied. In other words, the mechanism $\mathcal{M}(\bm{\theta}_j,\bm{z})$ used in QNN  perturbed by the cannel $\mathcal{E}_{p_1}$ is also DP. This observation promises that QNN  perturbed by the cannel $\mathcal{E}_{p_1}$ is a DP learning model.

\subsection{Learnability of QNN}\label{Appendix:subsec_Learn_QNN_general_channel}

\textbf{The generalization of Theorem \ref{thm:DP_QNNQAE_inform}.} Celebrated by the DP property of QNN as discussed above, we can effectively generalize the result of Theorem \ref{thm:DP_QNNQAE_inform} to the noisy channel $\mathcal{E}_{p_1}$ setting, i.e., if QNN with noiseless gates $\PAC$ learns a concept, then QNN perturbed by the noisy channel $\mathcal{E}_{p_1}$ can also learn such a concept using polynomial samples. 

\textbf{The generalization of Theorem \ref{thm:QNN_QAE_QSQ_info}.}  Analogous to the depolarization noise setting,   the distance between the target result $\nu=\Tr(\mathbb{M}\left({U}(\bm{\theta}) U_{\bm{x}}\right)\rho\left({U}(\bm{\theta}) U_{\bm{x}}\right)^{\dagger})$ and the shifted expectation value $\tilde{\nu}=(1-\tilde{p})\nu+p_2'\Tr(\mathbb{M}\kappa)+p_3^{L_Q}\Tr(\mathbb{M})/D$ of QNN under the noisy channel $\mathcal{E}_{p_1}$ follows $|\nu-\tilde{\nu}|\leq \tilde{p}\nu + p_2'+ p_3^{L_Q}/D$. Then by employing Chernoff-Hoeffding bound, we achieve, with probability at least $1-2\exp(-\delta^2 n/2)$, \[\left|\frac{1}{k}\sum_{k=1}^K V_k-\nu\right|\leq \tilde{p}\nu + p_2'+ \frac{p_3^{L_Q}}{D} + \frac{\delta}{2}~.\]    

Assuming $\tilde{p}\leq \frac{\tau- p_2'- \frac{p_3^{L_Q}}{D} - \frac{\delta}{2}}{\nu}$, with setting $\delta=2(\tau - \tilde{p}\nu - p_2'- {p_3^{L_Q}}/{D})$, the relation between the number of measurements $K$ and the successful probability $b$ obeys  
 \begin{equation}
 	\Pr\left(\left|\frac{1}{K}\sum_{k=1}^K V_k - \tilde{\nu} \right| \geq 
\left(\tau - \tilde{p}\nu - p_2'- \frac{p_3^{L_Q}}{D} \right) \right) \leq 2 \exp\left(-2{\left(\tau - \tilde{p}\nu - p_2'- \frac{p_3^{L_Q}}{D}\right)^2 K} \right)=b~.
 \end{equation}

After simplification, we conclude that, when $\tilde{p}\leq \frac{\tau- p_2'- \frac{p_3^{L_Q}}{D} - \frac{\delta}{2}}{\nu}$, with the successful probability at least $1-b$,  the required number of measurements  to attain $\left|\frac{1}{K}\sum_{k=1}^K V_k - \nu \right| \leq  \tau$ is
\begin{equation}
	K =  \frac{\ln\left(\frac{2}{b}\right)}{2 \left(\tau - \tilde{p}\nu - p_2'- \frac{p_3^{L_Q}}{D}\right)^2}~.
\end{equation}

\end{document}